\documentclass[lettersize,journal,twoside]{IEEEtran}
\pdfoutput=1
\IEEEoverridecommandlockouts
\usepackage{cite}
\usepackage{amsthm}
\usepackage{amsmath,amssymb,amsfonts}
\usepackage{graphicx} 
\usepackage{float} 
\usepackage{epstopdf}
\usepackage{caption}
\usepackage{algorithmic}
\usepackage{textcomp}
\usepackage{xcolor}
\usepackage{subfigure}
\usepackage[justification=centering]{caption}

\usepackage{booktabs}
\usepackage{multirow}
\usepackage{tabularx}
\usepackage{longtable}
\usepackage{supertabular}
\def\BibTeX{{\rm B\kern-.05em{\sc i\kern-.025em b}\kern-.08em
    T\kern-.1667em\lower.7ex\hbox{E}\kern-.125emX}}

\newtheorem{definition}{Definition}

\newtheorem{corollary}{Corollary}
\newtheorem{lemma}{Lemma}

\newtheorem{remark}{Remark}

\begin{document}

\title{FOGNA: An effective Sum-Difference Co-Array Design Based  on  Fourth-Order Cumulants}
\author{Si Wang and Guoqiang Xiao
\thanks{
Manuscript received xx December 2024.
This work was supported in part by the National Natural Science Foundation of China under Grant 62371400.
The associate editor coordinating the review of this manuscript and approving it for publication was xxx.
 (Corresponding author: Guoqiang Xiao.)}
\thanks{Si Wang and Guoqiang Xiao are with the College of Computer and Information Science, Southwest University, Chongqing 400715, China
 (e-mail:hw8888@email.swu.edu.cn; gqxiao@swu.edu.cn).}
}

\markboth{ FOGNA: An effective Sum-Difference Co-Array Design Based  on  Fourth-Order Cumulants,~Vol.~xx, No.~x, December~2024}%
{WNAG \MakeLowercase{\textit{et al.}}: FOGNA: An effective Sum-Difference Co-Array Design Based  on  Fourth-Order Cumulants}

\maketitle

\begin{abstract}
Array structures based on the fourth-order difference co-array (FODCA) provide more degrees of freedom (DOF).
However, since the growth of DOF is limited by a single case of fourth-order cumulant in FODCA,
this paper aims to design a sparse linear array (SLA) with higher DOF via
exploring different cases of fourth-order cumulants.
We present a mathematical framework based on fourth-order cumulant to devise a fourth-order extend co-array (FOECA),
which is equivalent to FODCA.
Based on FOECA, a novel fourth-order generalized nested array (FOGNA) is proposed in the paper,
which can provide closed-form expressions for the sensor
locations and enhance DOF in order to resolve more signal sources in the estimation of direction of arrival (DOA) .
FOGNA is consisted of three subarrays,
where the first is a concatenated nested array and the other two subarrays are ULAs with big inter-spacing between sensors.
When the total physical sensors are given, the number of sensors in each subarray of FOGNA  is determined by
the designed algorithm, which can obtain the maximum DOF under the proposed array structure and derive closed-form expressions for the sensor locations of FOGNA.
The proposed FOGNA not only achieves higher DOF than those of existing FODCAs
but also reduces mutual coupling effects.
Numerical simulations are conducted to verify the superiority of FOGNA on DOA estimation performance and enhanced DOF over other existing FODCAs.
\end{abstract}



\section{Introduction}
\IEEEPARstart{L}{ow} cost sampling in intelligent perception is widely applied in many fields such as frequency estimation in the time domain \cite{Xiao2017notes}, \cite{Xiao2018robustness}, \cite{Xiao2016symmetric}, \cite{Xiao2021wrapped}, \cite{Xiao2023on}
and DOA estimation in the spatial domain.
This paper mainly focuses on DOA estimation in array signal processing,
which is a fundamental problem studied for several decades \cite{Krim1996}, \cite{Godara1997}, \cite{Tuncer2009}, \cite{Xiao22023}.
It is well known that a uniform linear array (ULA) with $N$-sensors is used to estimate ($N-1$) sources using
DOA estimation methods such as MUSIC \cite{Schmidt1986} or ESPRIT \cite{Roy1989}.
To increase the DOF of ULA, more sensors are required, thus leading to a higher cost in practical applications.
ULAs also suffer from severe mutual coupling effects among physical sensors.
However, nonuniform linear arrays (also known as SLAs ) offer an effective solution to these problems.
For an $N$-sensors sparse array,
the corresponding difference co-array (DCA) can be constructed with the second-order cumulant (SOC) of the received signals,
which can provide $\mathcal{O}(N^2)$ consecutive lags \cite{Pal2010}, \cite{Pal2011}.
In this way, DOFs can be increased significantly compared to traditional ULAs,
while mutual coupling effects may also be reduced due to more larger inter-spacing between two sensors in SLAs \cite{BD2017mutual}.



%

Minimum redundancy array (MRA) \cite{Moffet1968} is a foundational structure in order to obtain as large DOF as possible
by reducing redundant sensors.
However, the non-closed form expressions for sensor positions hinder MRA's scalability and complicate large-scale array design.
This limitation leads to extensive research on nested arrays \cite{Pal2010} and coprime arrays \cite{Pal2011},
which offer significant advantages with their closed-form expressions for the sensor locations.
The success of nested arrays and coprime arrays inspires further developments aimed at enhancing DOF,
including augmented coprime arrays \cite{Pal22011},
enhanced nested arrays \cite{Zhao2019} and arrays based on the maximum element spacing criterion \cite{Zheng2019}.
With the respect of DOA estimation,
traditional subspace-based methods \cite{Liu2015} only utilize the consecutive lags of the DCA, making a hole-free configuration advantageous.
Consequently, hole-filling strategies have been proposed to create new coprime arrays-like with a hole-free DCA \cite{Wang2019}.
Obviously, various DCAs based on SOC have been widely studied in DOA estimation because of its significantly enhanced DOF \cite{Pal2010}, \cite{Pal2011}, \cite{Zhao2019}, \cite{Zheng2019}, \cite{Shi2022}.

Moreover, exploiting third-order cumulant (TOC) \cite{Sharma2023} or fourth-order cumulant (FOC) of SLAs
can further enhance DOF \cite{Xiao2023}, \cite{Guo2024}, starting from a mathematical model perspective.
Specifically, the corresponding third-order difference co-array (TODCA) and FODCA of SLA can be obtained by calculating the TOC and FOC of the signals received based on SLA separately.
At this time, an $N$-sensors sparse array provides $\mathcal{O}(N^3)$ and $\mathcal{O}(N^4)$ consecutive lags for TODCAs and FODCAs, respectively.

However, the TODCA requires an odd number of physical sensors,
which makes it more difficult to determine the physical sensor positions during array design.
As a result, the application of FODCA based on FOC in DOA estimation has recently attracted great attention \cite{Ahmed2017} \cite{Zhou2020}, \cite{LiuJ2017},
as it provides higher DOF to resolve more signal sources and obtains closed-form expressions of physical sensor positions more possibly.

In recent years, numerous representative studies have focused on the exploration of FODCA.
In \cite{Shen2016}, FODCA based on two-level nested arrays is studied.
The extension of four-level nested array is applyed to design FODCA in \cite{Piya2012},
which is an extension of the two-level nested array in \cite{Shen2016}.
Furthermore, the simplified and enhanced four-level nested array is proposed in \cite{Shen2019}, further increasing the DOF.
In addition, considering sum co-array (SCA), a new fourth-order sparse array called sum-difference-FODC is proposed in \cite{Guo2024},
which is composed of one sparse array and another extended-spacing sparse array.
By exploiting second-order SCA and DCA, more consecutive lags of FODCA are achieved.


Although SLAs designed based on FODCA can greatly increase the number of  consecutive lags,
using a single FOC case to design SLAs still limits the increase in the number of consecutive lags for SLAs.
Therefore, in order to increase the number of consecutive lags for virtual array obtained by using FOC to a greater extent,
a framework of fourth-order extend coarray (FOECA) is proposed in the paper with combining three cases of FOC for received signals.
Furthermore, based on FOECA, we propose a fourth-order generalized nested array (FOGNA) with  hole-free co-arrays and closed form expressions of sensor positions.
The proposed co-array is inspired by the ideas in \cite{Guo2024} and \cite{Shen2016},
where  SLAs in \cite{Guo2024} designed by using second-order DCA and SCA can greatly enhance DOF,
while the consecutive lags obtained by hole-free FODCA of extended nested array in \cite{Shen2016} can be increased significantly by adding a third ULA to the array structure.
Consequently, these approaches are considered in the specific case design of FOGNA to enhance the DOF.

In addition, among the commonly used subspace-based methods to estimate DOA, such as SS-MUSIC,
only the consecutive lags of DCA can by used \cite{Piya2012}.
Therefore, the number of resolvable sources in DOA estimation is greatly affected by the DOF\cite{Piya2012}, \cite{LiuCL2016},
which are commonly adopted as a indicator for quantitative evaluation and performance optimization of designing SLAs \cite{Shen2016}, \cite{Shen2019},\cite{Shen2015}.
Furthermore, to design a SLA structure with enhanced DOF by using different cases of FOC,
we introduce the criterion of forming more consecutive lags of designing SLA as follows,

\emph{Criterion 1 (Large consecutive lags of DCA)}:
The large consecutive lags of DCA are preferred, which can not only increase the number of resolvable sources
but also lead to higher spatial resolution in the DOA estimation \cite{Pal2010}, \cite{Piya2012}, \cite{Shen2019}.

\emph{Criterion 2 (Closed-form expressions of sensor positions)}:
A closed-form expression of sensor positions is preferred for scalability considerations \cite{Cohen2019}, \cite{Cohen2020}.

\emph{Criterion 3 (Hole-free DCA)}:
A SLA with a hole-free DCA is preferred,
since the data from its DCA can be utilized directly by subspace-based DOA estimation methods
which are easy to be implemented, and thus the algorithms based on compressive sensing \cite{Shen2015}, \cite{Shen20152}, \cite{Shen2017} or co-array interpolation techniques \cite{Cui2019}, \cite{Zhou2018} with increased computational complexity can be avoided \cite{Cohen2020}, \cite{Liu20172}.

\emph{Contribution:}
This paper focuses on the design of SLA in order to get hole-free FOECA based on FOC with enhance DOF.
The main contributions of the paper are threefold.\par
$\bullet$ In this paper, an effective framework of FOECA is devised mathematically based on three cases of FOC for non-Gaussian sources,
which can provide higher DOF than those of other SLAs  designed based on FODCA.
\par
$\bullet$ A novel FOGNA is systematically designed based on FOECA, utilizing three linear subarrays placed side-by-side,
which maximizes the consecutive lags of FOECA to enhance resolving the number of sources.
For the proposed FOGNA with the given number of physical sensors,
the closed-form expressions of the physical sensor positions have been derived analytically in this paper,
and the DOF of FOGNA is further enhanced by improving the configuration of the physical sensors among the three subarrays.
Consequently, the proposed FOGNA offers significantly higher DOF than those of other existing similar SLAs \cite{Guo2024}, \cite{Piya2012}, \cite{Shen2019}, \cite{Yang2023}.
\par
$\bullet$ Numerical simulations show that FOGNA outperforms other existing FODCAs in terms of resolution for DOA estimation.
Specifically, when the difference of incidence angles for two sources is reduced to $1.6^{\circ}$,
it can be observed that only FOGNA is able to resolve the two sources, while the other four arrays cannot.
Additionally, FOGNA performs better in DOA estimation under the conditions of with or without mutual coupling effects compared to other existing FODCAs,
as confirmed by numerical simulations.

This paper is organized as follows. In Section II, we briefly introduce the general sparse array model and the mutual coupling model.
In Section III, a mathematical framework is presented to derive a FOECA associated with three cases of FOC,
with further consideration of mutual coupling for the FOECA.
In Section IV, we describe how to design the proposed FOGNA based on FOECA and provide closed-form mathematical expressions for the physical sensors in FOGNA by using SCA and DCA.
Furthermore, we explain that the FOECA of proposed FOGNA is hole-free and calculate the corresponding maximum DOF.
In Section V, to illustrate the superior performance of FOGNA in enhancing DOF and reducing coupling effects,
we compare the DOF and coupling leakage of FOGNA with those of the other four FODCAs.
Further, several numerical simulations are presented to evaluate the RMSE of FOGNA
compared to the other FODCAs with respect to SNR, snapshots and the number of sources, both with and without considering coupling effects.

\textit{Notations:}
$\mathbb{S}$ is the physical sensor positions set of a SLA.
$N$ is the number of sensors. $D$ is the number of source signals to be estimated. $K$ is the number of snapshots.
$\Phi$ is sensor positions set of a FOECA.
$\mathbb{R}$, $\mathbb{C}$ and $\mathbb{Z}$  are the real number field, the complex number field and the integer number field, respectively.
The operators $\otimes$, $\odot$, $(\cdot)^T$, $(\cdot)^H$ and $(\cdot)^*$ stand for the Kronecker products,
Khatri-Rao products,
transpose, conjugate transpose and complex conjugation, respectively.
Set $\{a : b : c \}$ denotes the integer line from $a$ to $c$ sampled in steps of $b\in \mathbb{N}^+$.
When $b=1$, we use shorthand $\{a : c \}$ .

\section{Preliminaries}

\subsection{General Sparse Array Model}
Assume that there are $D$ non-Gaussian and mutually uncorrelated far-field narrow band sources.
The incident angle of the $i^{th}$ source is $\theta_i$,
and the physical sensor position set of the SLA is represented as $\mathbb{S}=\{ p_1,p_2,...,p_N \}\cdot d$,
where the unit spacing $d$ is generally set to half wavelength.
The output of the $n^{th}$ physical sensor at the $t^{th}$ snapshot,
denoted as $x_n(t)$, can be expressed as follows
\begin{equation}
\label{w1}
\begin{aligned}
x_n(t)=\sum_{i=1}^Da_n(\theta_i)s_i(t)+n_n(t),
\end{aligned}
\end{equation}
where $n_n(t)$ denotes a zero-mean additive Gaussian noise sample at the $n^{th}$ physical sensor, which is assumed to be
statistically independent of all the sources.
And $a_n(\theta_i)$ denotes the steering response of $n^{th}$ physical sensor corresponding to the $i^{th}$ source,
which can be expressed as follows
\begin{equation}
\label{w8}
a_n(\theta_i)=e^{j\frac{2\pi p_{l_n} d}{\lambda}\sin(\theta_i)}.
\end{equation}

Thus, the received signals for all the $N$ physical sensors are denoted as
$\boldsymbol{x}(t)=[x_1(t),...,x_N(t)]^T$, which can be written as follows
\begin{equation}
\label{w2}
\boldsymbol{x}(t)=\sum_{i=1}^D\boldsymbol{a}(\theta_i)s_i(t)+\boldsymbol{n}(t)=\boldsymbol{A}(\theta)\boldsymbol{s}(t)+\boldsymbol{n}(t),
\end{equation}
where $\boldsymbol{s}(t)=[s_1(t),...,s_D(t)]^T$ denotes the source signal vector,
$\boldsymbol{n}(t)=[n_1(t),...,n_N(t)]^T$ denotes the additive Gaussian noise vector,
and $\boldsymbol{a}(\theta_i)=[a_1(\theta_i),...,a_N(\theta_D)]^T$ denotes the array steering vector corresponding to the $i^{th}$ source,
and $\boldsymbol{A}(\theta)=[\boldsymbol{a}(\theta_1),...,\boldsymbol{a}(\theta_D)]\in\mathbb{C}^{N\times D}$
denotes the array manifold matrix. Set $p_1 = 0$, $\boldsymbol{a}(\theta_i)$ can be written as
\begin{equation}\nonumber
\boldsymbol{a}(\theta_i)=[1,e^{-j\frac{2\pi p_2d\sin\theta_i}{\lambda}},...,e^{-j\frac{2\pi p_Nd\sin\theta_i}{\lambda}}]^T.
\end{equation}

For $K$ numbers of snapshots, (\ref{w2}) can be rewritten in matrix form as follows
\begin{equation}
\boldsymbol{X}=\boldsymbol{A}\boldsymbol{S}+\boldsymbol{N},
\end{equation}
where $\boldsymbol{X}=[\boldsymbol{x}(1),..., \boldsymbol{x}(D)] \in\mathbb{C}^{N\times D}$ is the received signal matrix,
$\boldsymbol{S}=[\boldsymbol{s}(1),..., \boldsymbol{s}(D)] \in\mathbb{C}^{D\times K}$ is the source signal matrix,
and $\boldsymbol{N}=[\boldsymbol{n}(1),..., \boldsymbol{n}(D)] \in\mathbb{C}^{N\times K}$ is the additive Gaussian noise matrix.

Before describing the proposed array structure, we firstly introduce the definitions  of SCA and DCA  for the completeness of this paper.

\begin{definition}
 (SCA\cite{Robin2017}): Given a physical sensor position set, $\mathbb{S }=\{p_{l_1}, p_{l_2},..., p_{l_N}\}\cdot d$, of a  SLA,
 the following set determines the sensor positions of a SCA
  \begin{equation}
 \Phi_{\sum}=\{ (p_{l_1}+p_{l_2}) d \ | \ l_1,l_2=1,...,N\},
 \label{wang2}
 \end{equation}
 where $p_{l_1}$ and $p_{l_2}$ represent the physical sensor positions given by the set $\mathbb{S}$.
 \end{definition}

\begin{definition}

(DCA\cite{Dias2017}): Given a physical sensor position set, $\mathbb{S} =\{p_{l_1}, p_{l_2},..., p_{l_N}\}\cdot d$, of a  SLA,
the following set determines the sensor positions of a DCA
\begin{equation}
\Psi=\{ (p_{l_1}-p_{l_2}) d \ | \ l_1, l_2=1,2,...,N \}.
\end{equation}
where $p_{l_1}$ and $p_{l_2}$ represent the physical sensor positions given by the set $\mathbb{S}$.
\end{definition}

In addition, to lighten the notations, given any two sets $\mathbb{S}$ and $\mathbb{S}'$ \cite{Xiao2023}, we use
\begin{equation}
\mathbb{C}(\mathbb{S},\mathbb{S}')=\{ p_i+p_j\ | \ p_i\in \mathbb{S},p_j\in \mathbb{S}' \},
\end{equation}
to denote the cross sum of elements from $\mathbb{S}$ and $\mathbb{S}'$.

\subsection{Mutual Coupling Model}
There exists the mutual coupling effect among the physical sensors in practical applications.
When considering the effect of mutual coupling,
the received signal vector in (\ref{w2}) can be rewritten as
\begin{equation}
\label{wang5}
\boldsymbol{x}(t)=\boldsymbol{C}\boldsymbol{A}(\theta)\boldsymbol{s}(t)+\boldsymbol{n}(t),
\end{equation}
where $\boldsymbol{C}$ is the $N\times N$ mutual coupling matrix. Note that the coupling-free model in (\ref{w2}) can be regarded as a special case of (\ref{wang5}),
where $\boldsymbol{C}$ is an identity matrix.

In general, the expression for $\boldsymbol{C}$ is rather complicated\cite{LiuCL2016,LiuJ2017}. In the ULA configuration, $\boldsymbol{C}$ can be approximated by a
B-banded symmetric Toeplitz matrix as follows\cite{Friedlander1991}, \cite{Svantesson19991}, \cite{Svantesson19992}, \cite{YeZ2009}, \cite{Liao2012},
\begin{equation}
\boldsymbol{C}(n_1,n_2)=
\begin{cases}
c_{|n_1-n_2|},\ \ \ \ if\ |n_1-n_2|\leq B,\\
0,\ \ \ \ \ \ \ \ otherwise,
\end{cases}
\end{equation}
where $n_1,n_2\in\mathbb{S}$, and $c_0,c_1,...,c_B$ are coupling coefficients satisfying $c_0=1>|c_1|>|c_2|>...>|c_B|$ \cite{Friedlander1991}.

\begin{definition}
(Coupling Leakage \cite{Zheng2019}): For a given array with N-senors, the coupling leakage is defined as the energy ratio:
\begin{equation}
\label{w24}
L=\frac{\|\boldsymbol{C}-diag(\boldsymbol{C}) \|_F}{\|\boldsymbol{C} \|_F},
\end{equation}
where $\|\boldsymbol{C}-diag(\boldsymbol{C}) \|_F$ is the energy of all the off-diagonal components, which characterizes the level of mutual coupling.
A small value of $L$ implies that the mutual coupling is less significant.
\end{definition}

\section{FOURTH-ORDER EXTENDED CO-ARRAY}
Higher-order cumulants belong to higher-order statistics \cite{Yang2023}, which are used to describe the mathematical characteristics of stochastic processes.
When the received signals is not Gaussian, relying solely on second-order statistics cannot fully capture the statistical characteristics of the signals.
Therefore, we use higher-order cumulants, specifically FOC, to perform DOA estimation. For the zero-mean random process $\boldsymbol{x}(t)$,
the calculation expression for its FOC is as follows
\begin{equation}
\begin{aligned}
\mathcal{C}_{4,\boldsymbol{x}}(l_1,l_2,l_3,l_4)&=cum\{ \boldsymbol{x}_{l_1}(t),\boldsymbol{x}_{l_2}(t),\boldsymbol{x}_{l_3}(t),\boldsymbol{x}_{l_4}(t) \}\\
&=E\{\boldsymbol{x}_{l_1}(t)\boldsymbol{x}_{l_2}(t)\boldsymbol{x}_{l_3}(t)\boldsymbol{x}_{l_4}(t) \}\\
&-E\{\boldsymbol{x}_{l_1}(t)\boldsymbol{x}_{l_3}(t)\}E\{\boldsymbol{x}_{l_2}(t)\boldsymbol{x}_{l_4}(t) \}\\
&-E\{\boldsymbol{x}_{l_1}(t)\boldsymbol{x}_{l_4}(t)\}E\{\boldsymbol{x}_{l_2}(t)\boldsymbol{x}_{l_3}(t) \}\\
&-E\{\boldsymbol{x}_{l_1}(t)\boldsymbol{x}_{l_2}(t)\}E\{\boldsymbol{x}_{l_3}(t)\boldsymbol{x}_{l_4}(t) \},
\end{aligned}
\end{equation}
where $E\{\cdot\}$ represents the expectation operator. If the random process $\boldsymbol{x}(t)$ is a zero-mean Gaussian process, its FOC is equal to zero.

\subsection{Fourth-Order Extended Co-Array}
The FOC of received signal vector $\boldsymbol{x}(t)$ in (\ref{w2}) can be represented as the following three cases

\begin{subequations}
\label{eq1}
\begin{equation}
\begin{aligned}
\label{eq1a}
\mathcal{C}_{\boldsymbol{x}}^{(1)}&=cum\{\boldsymbol{x}(t),\boldsymbol{x}(t),\boldsymbol{x}^*(t),\boldsymbol{x}(t)\}\in\mathbb{C}^{N\times N\times N\times N}\\
&=\sum_{i=1}^DC_{4,s_i(t)}^{(1)}[\boldsymbol{a}(\theta_i)\otimes\boldsymbol{a}(\theta_i)]\times[\boldsymbol{a}^*(\theta_i)\otimes\boldsymbol{a}(\theta_i)]^H,\\
C_{4,s_i(t)}^{(1)}&=cum\{s_i(t),s_i(t),s_i^*(t),s_i(t)\},
\end{aligned}
\end{equation}

\begin{equation}
\begin{aligned}
\label{eq1b}
\mathcal{C}_{\boldsymbol{x}}^{(2)}&=cum\{\boldsymbol{x}(t),\boldsymbol{x}^*(t),\boldsymbol{x}^*(t),\boldsymbol{x}(t)\}\in\mathbb{C}^{N\times N\times N\times N}\\
&=\sum_{i=1}^DC_{4,s_i(t)}^{(2)}[\boldsymbol{a}(\theta_i)\otimes\boldsymbol{a}^*(\theta_i)]\times[\boldsymbol{a}(\theta_i)^*\otimes\boldsymbol{a}(\theta_i)]^H,\\
C_{4,s_i(t)}^{(2)}&=cum\{s_i(t),s_i^*(t),s_i^*(t),s_i(t)\},
\end{aligned}
\end{equation}

\begin{equation}
\begin{aligned}
\label{eq1c}
\mathcal{C}_{\boldsymbol{x}}^{(3)}&\triangleq cum\{\boldsymbol{x}^*(t),\boldsymbol{x}^*(t),\boldsymbol{x}(t),\boldsymbol{x}^*(t)\}\in\mathbb{C}^{N\times N\times N\times N}\\
&=\sum_{i=1}^DC_{4,s_i(t)}^{(3)}[\boldsymbol{a}^*(\theta_i)\otimes\boldsymbol{a}^*(\theta_i)]\times[\boldsymbol{a}(\theta_i)\otimes\boldsymbol{a}^*(\theta_i)]^H,\\
C_{4,s_i(t)}^{(3)}&=cum\{s_i^*(t),s_i^*(t),s_i(t),s^*_i(t)\},
\end{aligned}
\end{equation}
\end{subequations}
\textit{\textbf{Note:} For the cases of FOC,
there are a total of six cases, and the above three are selected to derive the FOECA.}

The FOCs in (\ref{eq1}) eliminate the corresponding noise $\boldsymbol{n}(t)$,
as the cumulants of order greater than 2 are identically zero for Gaussian random processes.
The vectorization of FOC $\mathcal{C}_{\boldsymbol{x}}^{(j)}$ generates the corresponding column vector $\boldsymbol{c}_{\boldsymbol{x}}^{(j)}\in\mathbb{C}^{N^4\times1}, j\in\{1,2,3\}$,
as follows

\begin{equation}
\label{wh1}
\begin{aligned}
\boldsymbol{c}_{\boldsymbol{x}}^{(j)}=vec(\mathcal{C}_{\boldsymbol{x}}^{(j)})=\sum_{i=1}^Db^{(j)}(\theta_i)p_{s_i}^{(j)}=\boldsymbol{B}^{(j)}\boldsymbol{p}_{\boldsymbol{s}}^{(j)},
\end{aligned}
\end{equation}
where    $\boldsymbol{b}^{(j)}(\theta_i)\in \mathbb{C}^{N^4\times1}$, $\boldsymbol{B}^{(j)}\in\mathbb{C}^{N^4\times D}$,
$p_{s_i}^{(j)}\in\mathbb{C}$, $\boldsymbol{p}_{\boldsymbol{s}}^{(j)}\in \mathbb{C}^{D\times 1}$ .

Further, a new FOECA can be derived as follows by combining three $\boldsymbol{c}_{\boldsymbol{x}}^{(j)},j\in\{1,2,3\}$,
whose DOF are significantly enhanced compared to FODCA

\begin{equation}
\label{w18}
\boldsymbol{c}_{\boldsymbol{x}}=[{\boldsymbol{c}_{\boldsymbol{x}}^{(1)}}^T,{\boldsymbol{c}_{\boldsymbol{x}}^{(2)}}^T,{\boldsymbol{c}_{\boldsymbol{x}}^{(3)}}^T]^T
\triangleq\boldsymbol{B}\boldsymbol{p}_s\in\mathbb{C}^{3N^4\times1},
\end{equation}
where the equivalent source signal vector $\boldsymbol{p}_s$ is expressed as follows

\begin{equation}
\label{w22}
\boldsymbol{p}_s\triangleq[{\boldsymbol{p}_s^{(1)}}^T,{\boldsymbol{p}_s^{(2)}}^T,{\boldsymbol{p}_s^{(3)}}^T]^T\in\mathbb{C}^{3D\times1},
\end{equation}
and the equivalent array manifold matrix $\boldsymbol{B}$ is
\begin{equation}
\label{w21}
\boldsymbol{B}= \left(
               \begin{matrix}
                 \boldsymbol{B}^{(1)} & 0                    & 0                    \\
                 0                    & \boldsymbol{B}^{(2)} &0                     \\
                 0                    & 0                    &\boldsymbol{B}^{(3)}   \\
               \end{matrix}
             \right),
\end{equation}
where the specific expression of $\boldsymbol{B}^{(j)},\{j=1,2,3\}$ in (\ref{w21}) is as follows
\begin{equation}
\label{w5}
\begin{aligned}
&\boldsymbol{B}^{(j)}\triangleq[\boldsymbol{b}^{(j)}(\theta_1),\boldsymbol{b}^{(j)}(\theta_2),...,\boldsymbol{b}^{(j)}(\theta_D)],\\
&\boldsymbol{b}^{(j)}(\theta_i)\triangleq[{b}^{(j)}_1(\theta_i),{b}^{(j)}_2(\theta_i),...,{b}^{(j)}_{N^4}(\theta_i)]^T,(i=1,2,...,D)\\
&\ \ \ \ \ \ =\begin{cases}
[\boldsymbol{a}^*(\theta_i)\otimes\boldsymbol{a}(\theta_i)]^*\otimes[\boldsymbol{a}(\theta_i)\otimes\boldsymbol{a}(\theta_i)], j=1,\\
[\boldsymbol{a}(\theta_i)^*\otimes\boldsymbol{a}(\theta_i)]^*\otimes[\boldsymbol{a}(\theta_i)\otimes\boldsymbol{a}^*(\theta_i)],j=2,\\
[\boldsymbol{a}(\theta_i)\otimes\boldsymbol{a}(\theta_i)]^*\otimes[\boldsymbol{a}^*(\theta_i)\otimes\boldsymbol{a}(\theta_i)],j=3.
\end{cases}
\end{aligned}
\end{equation}

The FOCs for the three cases of the source are as follows
\begin{equation}
\label{w4}
\begin{aligned}
&\boldsymbol{p}_{\boldsymbol{s}}^{(j)}\triangleq[C_{4,s_1(t)}^{(j)},C_{4,s_2(t)}^{(j)},...,C_{4,s_D(t)}^{(j)}]^T,\\
&\begin{matrix}
C_{4,s_i(t)}^{(j)}=\\
(i=1,2,...,D)       \\
\end{matrix}
\begin{cases}
cum[s_i(t),s_i(t),s_i(t),s_i^*(t)],j=1,\\
cum[s_i(t),s_i^*(t),s_i^*(t),s_i(t)],j=2,\\
cum[s_i^*(t),s_i^*(t),s_i^*(t),s_i(t)],j=3.
\end{cases}
\end{aligned}
\end{equation}

For $\boldsymbol{c}_{\boldsymbol{x}}^{(j)}=\boldsymbol{B}^{(j)}\boldsymbol{p}_{\boldsymbol{s}}^{(j)}, j\in\{1,2,3\}$ given in (\ref{wh1}),
it can be observed that $\boldsymbol{c}_{\boldsymbol{x}}^{(j)}$ is the result of vectorizing the FOC $\mathcal{C}_{\boldsymbol{x}}^{(j)}$,
where $\boldsymbol{p}_{\boldsymbol{s}}^{(j)}$ represents the equivalent source signals vector given in (\ref{w4}),
$\boldsymbol{B}^{(j)}$ represents the equivalent manifold matrix and $\boldsymbol{b}^{(j)}(\theta_i)$ represents
the equivalent steering vector given in (\ref{w5}) corresponding to the source signal.

Consequently, three virtual co-arrays under the three different cases can be obtained from
$\boldsymbol{c}_{\boldsymbol{x}}^{(j)}=\boldsymbol{B}^{(j)}\boldsymbol{p}_{\boldsymbol{s}}^{(j)}, j\in\{1,2,3\}$,
namely first fourth-order co-array (FOCA$_1$), second fourth-order co-array (FOCA$_2$) and third fourth-order co-array (FOCA$_3$),
which are defined as follows.

\textbf{\textit{case 1:}} When $j=1$, we can get $\boldsymbol{b}^{(1)}(\theta_i)=\boldsymbol{a}(\theta_i)\otimes\boldsymbol{a}(\theta_i)\otimes\boldsymbol{a}(\theta_i)\otimes\boldsymbol{a}^*(\theta_i)$
from (\ref{w5}), and the elements of $\boldsymbol{b}^{(1)}(\theta_i)$ are given as follows
\begin{equation}
 \label{w11}
 \begin{aligned}
 &b^{(1)}_{N^3(l_1-1)+N^2(l_2-1)+N(l_3-1)+l_4}(\theta_i)\\
 &=a_{l_1}(\theta_i)a_{l_2}(\theta_i)a_{l_3}(\theta_i)a_{l_4}^*(\theta_i)\\
 &=e^{j\frac{2\pi d}{\lambda}(p_{l_1}+p_{l_2}+p_{l_3}-p_{l_4})\sin(\theta_i)}.
 \end{aligned}
 \end{equation}

Compared with the steering response $a_n(\theta_i)=e^{j\frac{2\pi p_{l_n} d}{\lambda}\sin(\theta_i)}$ for ULAs,
(\ref{w11}) implys the steering response of  sensor located at $(p_{l_1}+p_{l_2}+p_{l_3}-p_{l_4}) d$ for FOCA$_1$.
Consequently, the FOCA$_1$ derived from the vector $\boldsymbol{c}_{\boldsymbol{x}}^{(1)}$ in (\ref{wh1})
can be considered as the equivalent virtual array to receive signals, which is defined as follows.

\begin{definition}
(FOCA$_1$): For a linear array with N-sensors located at positions given by the set $\mathbb{S}$, a multiset $\Phi_1$ is defined as follows
\begin{equation}
\label{w28}
\Phi_1\triangleq\{(p_{l_1}+p_{l_2}+p_{l_3}-p_{l_4})d\ | \ l_1,l_2,l_3,l_4=1,2,...,N\},
\end{equation}
where the multiset $\Phi_1$ allows repetitions, and has an underlying set $\Phi_1^{u}$ that contains the unique elements of $\Phi_1$.
Consequently, FOCA$_1$ is defined as the virtual linear array for case 1,
where the sensors are located at positions given by the set $\Phi_1^u$.
\end{definition}

\textbf{\textit{case 2:}} When $j=2$, $\boldsymbol{b}^{(2)}(\theta_i)=\boldsymbol{a}(\theta_i)\otimes\boldsymbol{a}^*(\theta_i)\otimes\boldsymbol{a}(\theta_i)\otimes\boldsymbol{a}^*(\theta_i)$
in (\ref{w5}) is obtained, and the elements of $\boldsymbol{b}^{(2)}(\theta_i)$ are given as follows
\begin{equation}
 \label{wh2}
 \begin{aligned}
 &b^{(2)}_{N^3(l_1-1)+N^2(l_2-1)+N(l_3-1)+l_4}(\theta_i)\\
 &=a_{l_1}(\theta_i)a_{l_2}^*(\theta_i)a_{l_3}(\theta_i)a_{l_4}^*(\theta_i)\\
 &=e^{j\frac{2\pi d}{\lambda}(p_{l_1}-p_{l_2}+p_{l_3}-p_{l_4})\sin(\theta_i)}.
 \end{aligned}
 \end{equation}

Compared with the steering response $a_n(\theta_i)=e^{j\frac{2\pi p_{l_n} d}{\lambda}\sin(\theta_i)}$ for ULAs,
(\ref{wh2}) implys the steering response of  sensor located at $(p_{l_1}-p_{l_2}+p_{l_3}-p_{l_4}) d$ for FOCA$_2$.
Consequently, the FOCA$_2$ derived from the vector $\boldsymbol{c}_{\boldsymbol{x}}^{(2)}$ in (\ref{wh1})
can be considered as the equivalent virtual array to receive signals, which is defined as follows.
\begin{definition}
(FOCA$_2$): For a linear array with N-sensors located at positions given by the set $\mathbb{S}$, a multiset $\Phi_2$ is defined as follows
\begin{equation}
\label{w27}
\Phi_2\triangleq\{(p_{l_1}-p_{l_2}+p_{l_3}-p_{l_4})d \ | \ l_1,l_2,l_3,l_4=1,2,...,N\},
\end{equation}
where the multiset $\Phi_2$ allows repetitions, and has an underlying set $\Phi_2^{u}$ that contains the unique elements of $\Phi_2$.
Consequently, FOCA$_2$ is defined as the virtual linear array for case 2,
where the sensors are located at positions given by the set $\Phi_2^u$.
\end{definition}

\textbf{\textit{case 3:}} When $j=3$, $\boldsymbol{b}^{(3)}(\theta_i)=\boldsymbol{a}^*(\theta_i)\otimes\boldsymbol{a}^*(\theta_i)\otimes\boldsymbol{a}^*(\theta_i)\otimes\boldsymbol{a}(\theta_i)$
in (\ref{w5}) is obtained, and the elements of $\boldsymbol{b}^{(3)}(\theta_i)$ are given as follows
\begin{equation}
 \label{wh3}
 \begin{aligned}
 &b^{(3)}_{N^3(l_1-1)+N^2(l_2-1)+N(l_3-1)+l_4}(\theta_i)\\
 &=a_{l_1}^*(\theta_i)a_{l_2}^*(\theta_i)a_{l_3}^*(\theta_i)a_{l_4}(\theta_i)\\
 &=e^{j\frac{2\pi d}{\lambda}(-p_{l_1}-p_{l_2}-p_{l_3}+p_{l_4})\sin(\theta_i)}.
 \end{aligned}
 \end{equation}

Compared with the steering response $a_n(\theta_i)=e^{j\frac{2\pi p_{l_n} d}{\lambda}\sin(\theta_i)}$ for ULAs,
(\ref{wh3}) implys the steering response of  sensor located at $(-p_{l_1}-p_{l_2}-p_{l_3}+p_{l_4}) d$ for FOCA$_3$.
Consequently, the FOCA$_3$ derived from the vector $\boldsymbol{c}_{\boldsymbol{x}}^{(3)}$ in (\ref{wh1})
can be considered as the equivalent virtual array to receive signals, which is defined as follows.

\begin{definition}
(FOCA$_3$): For a linear array with N-sensors located at positions given by the set $\mathbb{S}$, a multiset $\Phi_3$ is defined as follows
\begin{equation}
\label{w30}
\Phi_3\triangleq\{(-p_{l_1}-p_{l_2}-p_{l_3}+p_{l_4})d \ | \ l_1,l_2,l_3,l_4=1,2,...,N\},
\end{equation}
where the multiset $\Phi_3$ allows repetitions, and has an underlying set $\Phi_3^{u}$ that contains the unique elements of $\Phi_3$.
Consequently, FOCA$_3$ is defined as the virtual linear array for case 3,
where the sensors are located at positions given by the set $\Phi_3^u$.
\end{definition}

Furthermore, $\boldsymbol{c}_{\boldsymbol{x}}\in\mathbb{C}^{3N^4\times1}$ in (\ref{w18}) is obtained by
combing three $\boldsymbol{c}_{\boldsymbol{x}}^{(j)}, j\in\{1,2,3\}$.
It is equivalent to the cumulants of the signals received in a single snapshot by the constructed virtual linear array,
which is the combination of all three possible co-arrays for $j\in \{1, 2, 3\}$, namely FOCA$_1$, FOCA$_2$ and FOCA$_3$.
Therefore, the derived virtual linear array is the FOECA,
and to obtain the sensor position set of FOECA, we first introduce the following knowledge about multiset.

A multiset $\Phi$ is defined as the multiset-sum (bag sum) of the three multisets $\Phi_1$, $\Phi_2$ and $\Phi_3$ corresponding to the aforementioned three co-arrays, i.e.
\begin{equation}
\label{w19}
\Phi= \Phi_1 + \Phi_2 + \Phi_3.
\end{equation}

For the multiset-sum of two multisets $\Theta_1=(Q_1,m_1)$ and $\Theta_2=(Q_2,m_2)$,
where $Q_1$ and $Q_2$ are the base set of the multiset $\Theta_1$ and $\Theta_2$ respectively (i.e., the set of all possible elements),
and $m_1:Q_1\rightarrow\mathbb{N}$ and $m_2:Q_2\rightarrow\mathbb{N}$
represent functions that describe the relationships of the occurrence counts of all elements in $Q_1$ and $Q_2$ relative to $N$, respectively.
Thus, multiset-sum is denoted as $\Theta_1 + \Theta_2$, and its multiplicity function $m_+(\alpha)$ is defined as follows \cite{Hickman1980}
\begin{equation}
\begin{aligned}
&\Theta_1+\Theta_2=(Q_1\cup Q_2,m_+(\alpha)),\ \ \forall \alpha\in Q_1\cup Q_2,\\
&m_+(\alpha)=m_1(\alpha)+ m_2(\alpha),\ \ \forall \alpha\in Q_1\cup Q_2.
\end{aligned}
\end{equation}

For each element in the base set, its multiplicity in the multiset-sum is the sum of the multiplicities in the two multisets.

The multiset $\Phi$ in (\ref{w19}) allows repetitions, and has an underlying set $\Phi^u$ that contains the unique elements of $\Phi$, i.e.
\begin{equation}
\label{w31}
\Phi^u=\Phi_1^u \cup \Phi_2^u \cup \Phi_3^u,
\end{equation}
where  $\cup$ denotes the union operation of the set without duplicate elements,
$\Theta_1\cup \Theta_2\triangleq\{x | x\in \Theta_1\ or\ x\in \Theta_2\}$,
$|\Theta_1\cup \Theta_2|\triangleq|\Theta_1|+|\Theta_2|-|\Theta_1\cap \Theta_2|$ \cite{Haan2007}.\\
~

After obtaining the sensor position set of FOECA, the FOECA is defined as follows.
\begin{definition}
(FOECA):
For a linear array of N-sensors located at positions given by the set $\mathbb{S}$,
the FOECA is derived based on the FOCs, whose sensors are located at the set $\Phi^u$.
\end{definition}

Therefore, the FOECA is defined as the virtual linear array corresponding to $|\Phi^u|=\mathcal{O}(3N^4)$ sensors.
In addition, it worths to notes  the following two remarks in order to fully understand FOECA.


\begin{remark}
$\Phi_2$ in (\ref{w27}) is symmetric corresponding to zero.
$\Phi_1$ in (\ref{w28}) and $\Phi_3$ in (\ref{w30}) are opposite numbers to each other.
Consequently, the sensor position set $\Phi^u$ of FOECA in (\ref{w31}) is the union of the three sets $\Phi_1^u$, $\Phi_2^u$ and $\Phi_3^u$,
which is symmetric corresponding to zero.
It means that if a virtual sensor locate at $p \in \Phi^u\cdot d$, there must exist another corresponding virtual sensor located at $-p \in \Phi^u\cdot d$.
\end{remark}

\begin{remark}
In general, the FOECA of an arbitrary linear array might not be a hole-free array.
For example, the FOECA of a linear array with sensor positions given by $\mathbb{S} = \{0,1,5,8\}\cdot d$ can be obtained based on FOCs,
and the virtual sensor positions on the non-negative side are given by
$\{0,1,2,3,4,5,6,7,8,9,10,11,12,13,14,15,16,17,\text{X},19,\text{X},$ \\ $21,\text{X},23,24\}\cdot d$ without 18d, 20d, 22d.
\end{remark}

\subsection{The FOECA With Mutual Coupling}
%

The FOECA with mutual coupling can be derived based on the mutual coupling model in (\ref{wang5}) as follows
\begin{equation}
\label{eq2}
\begin{aligned}
\boldsymbol{\tilde{z}}_{\boldsymbol{x}}^{(1)}=\boldsymbol{C}^{(1)}_{vec}\boldsymbol{B}^{(1)}\boldsymbol{p}_{\boldsymbol{s}}^{(1)},\\
\boldsymbol{\tilde{z}}_{\boldsymbol{x}}^{(2)}=\boldsymbol{C}^{(2)}_{vec}\boldsymbol{B}^{(2)}\boldsymbol{p}_{\boldsymbol{s}}^{(2)},\\
\boldsymbol{\tilde{z}}_{\boldsymbol{x}}^{(3)}=\boldsymbol{C}^{(3)}_{vec}\boldsymbol{B}^{(3)}\boldsymbol{p}_{\boldsymbol{s}}^{(3)}.
\end{aligned}
\end{equation}

Further, combining the three $\boldsymbol{\tilde{z}}_{\boldsymbol{x}}^{(j)},j\in\{1,2,3\}$ to
derive FOECA with mutual coupling
\begin{equation}
\begin{aligned}
\label{w20}
\boldsymbol{\tilde{z}}_{\boldsymbol{x}}&=[{\boldsymbol{\tilde{z}}_{\boldsymbol{x}}^{(1)}}^T,{\boldsymbol{\tilde{z}}_{\boldsymbol{x}}^{(2)}}^T,
{\boldsymbol{\tilde{z}}_{\boldsymbol{x}}^{(3)}}^T]^T
\triangleq\boldsymbol{\tilde{C}}_{vec}\boldsymbol{B}\boldsymbol{p}_s\in\mathbb{C}^{3N^4\times1},
\end{aligned}
\end{equation}
where $\boldsymbol{B}$ and $\boldsymbol{p}_s$ are shown in (\ref{w21}) and (\ref{w22}), respectively.
And the virtual mutual coupling matrix $\boldsymbol{\tilde{C}}_{vec}\in\mathbb{C}^{3N^4\times3N^4}$ of FOECA is shown as follows
\begin{equation}
\begin{cases}
\boldsymbol{C}^{(1)}_{vec}=(\boldsymbol{C}^*\otimes\boldsymbol{C}^*)^*\otimes(\boldsymbol{C}\otimes\boldsymbol{C}^*),\\
\boldsymbol{C}^{(2)}_{vec}=(\boldsymbol{C}^*\otimes\boldsymbol{C})^*\otimes(\boldsymbol{C}\otimes\boldsymbol{C}^*),\\
\boldsymbol{C}^{(3)}_{vec}=(\boldsymbol{C}\otimes\boldsymbol{C})^*\otimes(\boldsymbol{C}^*\otimes\boldsymbol{C}),
\end{cases}
\end{equation}
\begin{equation}
\boldsymbol{\tilde{C}}_{vec}= \left(
               \begin{matrix}
                 \boldsymbol{C}^{(1)}_{vec} & 0                          & 0                         \\
                 0                          & \boldsymbol{C}^{(2)}_{vec} &0                          \\
                 0                          & 0                          &\boldsymbol{C}^{(3)}_{vec} \\
               \end{matrix}
             \right).
\end{equation}

\section{THE STRUCTURE OF FOURTH-ORDER GENERALIZED NESTED ARRAY}
When the FOECA is hole-free, it can be easily utilized to estimate DOA without any spatial aliasing \cite{Pal2010}, \cite{Pal22011}, \cite{Piya2012}.
Moreover, the number of consecutive lags of DCA mainly
depends on the physical sensor geometry of a linear array \cite{Cohen2020}.
Therefore, the FOGNA is proposed based on FOECA in this paper to enhance the DOF.
The FOGNA is designed systematically by appropriately deploying the physical sensor positions of three subarrays as shown in Fig. 1,
where subarray 1 is a concatenated nested array (CNA) with $N_1$ physical sensors,
and subarray 2 and 3 are a SLA with big inter-spacings among sensors.
\subsection{Structure of the Proposed FOGNA Based on FOECA}

\begin{definition}
(FOGNA): The FOGNA consists of three subarrays with the number of physical sensors $N=N_1+N_2+N_3$, where $N_1 (N_1\geq2)$, $N_2$
and $N_3$ represent the number of physical sensors in subarray 1, 2 and 3.
These sensors in FOGNA are located at positions given by the set $\mathbb{S}_1$, $\mathbb{S}_2$ and $\mathbb{S}_3$, respectively,
which can be represented as follows
\begin{equation}
\label{wh5}
\begin{aligned}
\mathbb{S}&=\mathbb{S}_1\cup\mathbb{S}_2\cup\mathbb{S}_3,\\
\mathbb{S}_1
&=\{ 0:M_1-1 \}\cdot d\ \cup\\
&\ \ \ \ \{ M_1:M_1+1:M_1+(M_1+1)(M_2-1) \}\cdot d\ \cup\\
&\ \ \ \ \{M_1+(M_1+1)(M_2-1)+1:\\
&\ \ \ \ 2M_1+(M_1+1)(M_2-1) \}\cdot d,\\
\mathbb{S}_2
&=\{ 4E_1+1:2E_1+1:2E_1+N_2(2E_1+1) \}\cdot d\\
\mathbb{S}_3
&=\{ 2E_2:2E_2:2N_3E_2\} \cdot d, \\
E_1&=-2(\lceil\frac{N_1-1}{4}\rfloor)^2+(N_1-1)\lceil\frac{N_1-1}{4}\rfloor+N_1-1,\\
E_2&=2E_1+N_2(2E_1+1),\\
M_1&=\lceil\frac{N_1-1}{4}\rfloor, \ M_2=N_1-2M_1=N_1-2\lceil\frac{N_1-1}{4}\rfloor,
\end{aligned}
\end{equation}
where $E_1$ and $E_2$ are the aperture of the subarray 1 and 2 respectively.
Subarray 1 is CNA and also composed of three ULAs,
where $M_1$ represents the number of physical sensors in the first and third ULAs of CNA,
and $M_2$ represents the number of physical sensors in the second ULA of CNA.
The structure of the FOGNA designed based on FOECA is shown as Fig. 1.
\end{definition}

\begin{figure*}
 \center{\includegraphics[width=17cm]  {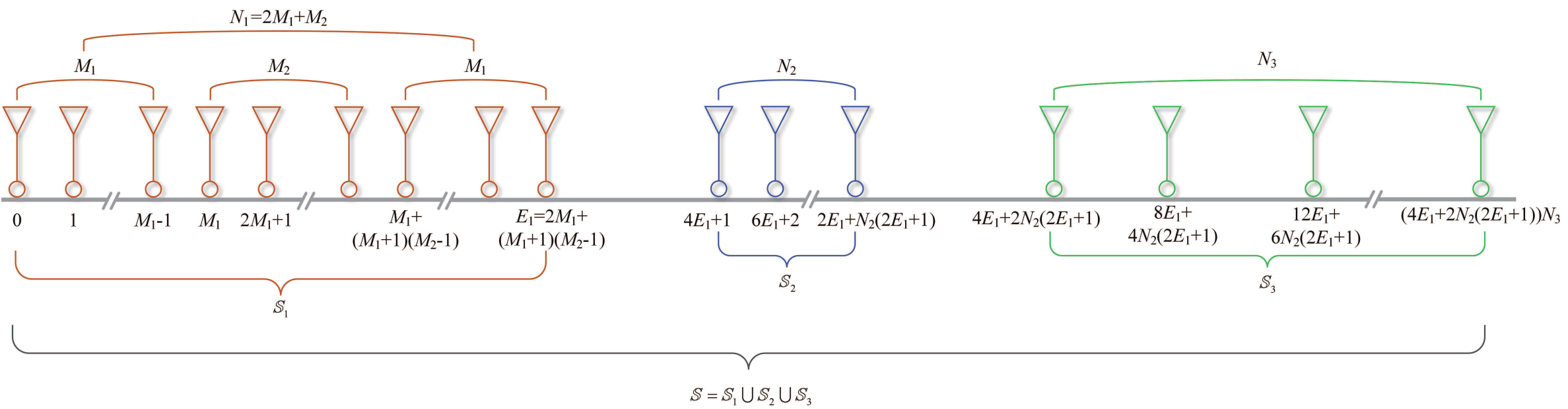}}
 \caption{\label{1} Structure of FOGNA}
\end{figure*}

\subsection{Consecutive Lags of the FOGNA}

\begin{lemma}
The FOECA of FOGNA is hole-free.
\end{lemma}

\begin{proof}
Firstly, the subarray 1 of FOGNA is a CNA with $N_1$ physical sensors,
which is formally expressed by $M_1,M_2\in \mathbb{N}$, and the physical sensor positions of the CNA are given by $\mathbb{S}_1$ in (\ref{wh5}).

The SCA of the CNA is the virtual array with $2E_1+1$ virtual sensors located at positions given by the set $\mathbb{V}_1=\{0:2E_1 \}\cdot d$.
From the expression of $\mathbb{V}_1$, it can be seen that the SCA is a hole-free array \cite{Robin2017}, which has $2E_1+1$ consecutive lags equivalently.

Secondly, for the subarray 2 of FOGNA, a ULA with $N_2$ physical sensors located at positions given by the set $\mathbb{S}_2$ in (\ref{wh5})
concatenating the SCA of subarray 1 can obtain another virtual array,
whose sensors are located at positions given by the set $\mathbb{V}_2$
\begin{equation}\nonumber
\begin{aligned}
\mathbb{V}_2&=\mathbb{V}_1\cup \mathbb{S}_2=\{ 0:2E_1 \}\cdot d\ \cup \\
&\ \ \ \ \{ 4E_1+1:2E_1+1:2E_1+N_2(2E_1+1) \}\cdot d.
\end{aligned}
\end{equation}

Furthermore, we can obtain the self-difference and cross-difference sets as follows
\begin{equation}\nonumber
\begin{aligned}
\mathbb{C}(\mathbb{V}_1,-\mathbb{S}_2)&=\{\mu d\ | -2E_1-N_2(2E_1+1) \leq\mu\leq -2E_1-1 \},\\
\mathbb{C}(-\mathbb{V}_1,\mathbb{S}_2)&=\{\mu d\ |\ 2E_1+1\leq\mu\leq 2E_1+N_2(2E_1+1) \},\\
\mathbb{C}(\mathbb{V}_1,-\mathbb{V}_1)&=\{\mu d\ | -2E_1\leq\mu\leq2E_1) \},
\end{aligned}
\end{equation}
and the sensor positions of virtual array are given by the following set
\begin{equation}
\label{wh6}
\begin{aligned}
&\mathbb{C}(\mathbb{V}_2,-\mathbb{V}_2)=\mathbb{C}(\mathbb{V}_1,-\mathbb{S}_2)
\cup\mathbb{C}(\mathbb{V}_1,-\mathbb{V}_1)\cup \mathbb{C}(-\mathbb{V}_1,\mathbb{S}_2)\\
&=\{\mu d\ |-2E_1-N_2(2E_1+1)\leq\mu \leq2E_1+N_2(2E_1+1) \}.
\end{aligned}
\end{equation}

Therefore, for the linear array with sensors located at positions given by the set $\mathbb{S}_1\cup \mathbb{S}_2$,
the virtual array can be obtained based on $(p_{l_1}+p_{l_2}-p_{l_3})d$ and $(-p_{l_1}-p_{l_2}+p_{l_3})d$,
where $p_{l_1},p_{l_2}\in \mathbb{S}_1,p_{l_3}\in \mathbb{S}_2$ \cite{Xiao2023}, whose sensors located at positions given
by the set $\mathbb{C}(\mathbb{V}_2,-\mathbb{V}_2)$ in (\ref{wh6}) are consecutive. That means the virtual array is hole-free.


Thirdly, for the subarray 3 of FOGNA, a ULA with $N_3$ physical sensors located at positions given by the set $\mathbb{S}_3$ in (\ref{wh5})
concatenating the virtual array of sensors located at positions given by the set $\mathbb{C}(\mathbb{V}_2,-\mathbb{V}_2)$ can obtain another virtual array
with sensors located at positions given by the set $\mathbb{V}_3$
\begin{equation}\nonumber
\begin{aligned}
\mathbb{V}_3&=\mathbb{C}(\mathbb{V}_2,-\mathbb{V}_2)\cup\mathbb{S}_3\\
&=\{-2E_1-N_2(2E_1+1):2E_1+N_2(2E_1+1)\}\cdot d \\
&\cup \{ 4E_1+2N_2(2E_1+1):4E_1+2N_2(2E_1+1):\\
&\ \ \ \ 4N_3E_1+2N_3N_2(2E_1+1)\} \cdot d.
\end{aligned}
\end{equation}

For simplifying the representation, the $\mathbb{C}(\mathbb{V}_2,-\mathbb{V}_2)$ is represented as $\Omega$,
the cross-difference sets can be obtained as follows
\begin{equation}\nonumber
\begin{aligned}
&\mathbb{C}(\Omega,-\mathbb{S}_3)=\{ \mu d \ |-2E_1(2N_3+1)-(2E_1+1)(N_2\\
&\ \ \ \ \ \ \ \ \ \ \ \ \ \ \ \ \ +2N_3N_2) \leq \mu \leq -2E_1-N_2(2E_1+1) \},\\
&\mathbb{C}(-\Omega,\mathbb{S}_3)=\{ \mu d \ |\ 2E_1+N_2(2E_1+1) \leq \mu \\
&\ \ \ \ \ \ \ \ \ \ \ \ \ \ \ \leq 2E_1(2N_3+1)+(2E_1+1)(N_2+2N_3N_2) \}.
\end{aligned}
\end{equation}

Further we can obtain the union of $\mathbb{C}(\Omega,-\mathbb{S}_3)$, $\mathbb{C}(-\Omega,\mathbb{S}_3)$ and $\mathbb{C}(\mathbb{V}_2,-\mathbb{V}_2)$ as follows
\begin{equation}
\label{wh7}
\begin{aligned}
\mathbb{C}(\mathbb{V}_3,-\mathbb{V}_3)&=\mathbb{C}(\Omega,-\mathbb{S}_3)\cup\mathbb{C}(\mathbb{V}_2,-\mathbb{V}_2)
\cup\mathbb{C}(-\Omega,\mathbb{S}_3)\\
&=\{\mu d \ | -(2N_3+1)(2E_1+N_2(2E_1+1))\\
&\ \ \ \ \leq \mu \leq (2N_3+1)(2E_1+N_2(2E_1+1))\}.
\end{aligned}
\end{equation}

Therefore, for the linear array with sensors located at positions given by the set $\mathbb{S}_1\cup \mathbb{S}_2\cup \mathbb{S}_3$,
the virtual array can be obtained based on $\pm(p_{l_1}+p_{l_2}-p_{l_3}-p_{l_4})d$ and $\pm(-p_{l_1}-p_{l_2}+p_{l_3}-p_{l_4})d$ ,
where $p_{l_1},p_{l_2}\in \mathbb{S}_1,p_{l_3}\in \mathbb{S}_2,p_{l_4}\in \mathbb{S}_3$, whose sensors located at positions given
by the set $\mathbb{C}(\mathbb{V}_3,-\mathbb{V}_3)$ in (\ref{wh7}) are consecutive. That means the virtual array is hole-free.

To sum up, the FOECA of proposed FOGNA is hole-free, and the total consecutive lags of FOECA are $2(2N_3+1)(2E_1+N_2(2E_1+1))+1$.
\end{proof}

\subsection{The Maximum DOF of FOGNA With the Given Number of Physical Sensors}
The DOF of FOGNA designed based on FOECA can be further increased by optimizing the distribution of the physical sensors
among subarray 1, 2 and 3 for a given number of physical sensors.
\begin{lemma}
To obtain the maximum DOF of FOGNA with the given number of physical sensors,
the number of sensors in the three subarrays is set to

\begin{equation}
\begin{cases}
N_1=Caculated\ by\ Algorithm\ 1,\\
N_2=\lceil\frac{2(N-N_1)-1}{4}\rceil,\\
N_3=\lfloor\frac{2(N-N_1)+1}{4}\rfloor.
\end{cases}
\end{equation}

For the subarray 1, in order to achieve maximum consecutive lags,
$M_1$ and $M_2$ in $N_1=2M_1+M_2$ are set to
\begin{equation}
\begin{cases}
M_1=\lceil\frac{N_1-1}{4}\rfloor,\\
M_2=N_1-2M_1=N_1-2\lceil\frac{N_1-1}{4}\rfloor.
\end{cases}
\end{equation}

\end{lemma}

\begin{proof}
To obtain the maximum DOF of FOGNA,
it is equivalent to solve the following optimization problem

\begin{equation}\nonumber
\begin{aligned}
\underset{N_1,N_2,N_3\in \mathbb{N}_+}{maximize}\ \ 2(2N_3+1)(2E_1+N_2(2E_1+1))+1,\\
subject\ to\ N_1+N_2+N_3=N. \ \ \ \ \ \ \ \ \ \ \ \ \ (P1)
\end{aligned}
\end{equation}

For problem (P1), it is known that the objective function is a cubic function.
Due to the change in sign of the second derivative (i.e. curvature) of a cubic function, the function curve may bend upwards or downwards,
making it a non-convex optimization problem.
In a non-convex optimization problem, multiple local extrema exist within the domain
of definition for independent variable, so a global optimum solution cannot be obtained directly \cite{Krentel1986}.
Therefore, we seek the optimal solution to problem (P1) by reducing unknown variables.
To be specific, assuming $N_1$ for subarray 1 is known, to achieve maximum consecutive lags of SCA obtained by CNA,
$E_1$ is also known at this time according to \cite{Robin2021}.
Therefore, the problem (P1) is changed into an optimization problem (P2) with only two unknown variables $N_2$ and $N_3$, which is expressed as follows

\begin{equation}\nonumber
\begin{aligned}
\underset{N_1,N_2,N_3\in \mathbb{N}_+}{maximize}\ \ 2(2N_3+1)(2E_1+N_2(2E_1+1))+1,\\
subject\ to\ N_2+N_3=N-N_1. \ \ \ \ \ \ \ \ \ \ \ \ \ \ (P2)
\end{aligned}
\end{equation}

For problem (P2), firstly, the maximum consecutive lags of subarray 1 can be derived when the number of physical sensors $N_1$ is known.
When subarray 1 is CNA, according to \cite{Robin2017}, \cite{Robin2021},
it is composed of three ULAs with symmetric properties, which is formally defined by
the given $M_1, M_2 \in \mathbb{N}$, and the inter-spacings of sensors in CNA are given by
\begin{equation}
 \mathcal{D}_{CNA}= \{ 1^{M_1},\ \  (M_1+1)^{(M_2-1)},\ \  1^{M_1} \}\cdot d,
\end{equation}
where, $a^b$ represents $a$ appearing $b$ times and $a$ represents the spacing between two physical sensors.
The aperture $E_1$ of the CNA and the number of physical sensors $N_1$ can be obtained as follows
\begin{equation}
\begin{cases}
N_1=2M_1+M_2,\\
E_1=2M_1+(M_1+1)(M_2-1).
\end{cases}
\end{equation}

To obtain the maximum DOF of CNA based on SCA, $M_1$ and $M_2$ are obtained as
\begin{equation}
\begin{cases}
M_1=\lceil\frac{N_1-1}{4}\rfloor,\\
M_2=N_1-2\lceil\frac{N_1-1}{4}\rfloor,
\end{cases}
\end{equation}
where $\lceil\cdot\rfloor$ denotes rounding to the nearest integer.
Substituting $M_1$ and $M_2$ into the aperture $E_1$ of the CNA yields
\begin{equation}
\label{w33}
\begin{aligned}
E_1=-2(\lceil\frac{N_1-1}{4}\rfloor)^2+(N_1-1)\lceil\frac{N_1-1}{4}\rfloor+N_1-1.
\end{aligned}
\end{equation}

At this time, it is a convex optimization problem for (P2),
which can directly yield the global optimum within the domain of definition for independent variable \cite{Krentel1986}.
To solving the solution of problem (P2), substituting $N_2=N-N_1-N_3$ into DOF yields

\begin{equation}
\begin{aligned}
f(N_3)&\triangleq
\text{DOF}&\\
&=2(2N_3+1)(2E_1+N_2(2E_1+1))+1\\
&=2[-2N_3^2(2E_1+1)+(4E_1+2(N-N_1)(2E_1+1)\\
&-(2E_1+1))N_3+2E_1+(N-N_1)(2E_1+1)]+1.
\end{aligned}
\label{wang3}
\end{equation}


If $N_1$ is known, $E_1$ is also known.
Therefore, when the total number of physical sensors $N$ is given,
the unknown parameter in (\ref{wang3}) is only $N_3$, which is a quadratic function $f(N_3)$ about $N_3$.
Furthermore, the solution to problem (P2) is to find the maximum value of the quadratic function $f(N_3)$ with respect to $N_3$.
As it is a convex optimization problem, according to the properties of a quadratic function,
it can be known that its maximum value is obtained at the first derivative of $f(N_3)$ equaling to zero \cite{Krentel1986}.
The first derivative of $f(N_3)$ is derived as follows

\begin{equation}
\label{wh4}
\begin{aligned}
&\frac{\partial f(N_3)}{\partial N_3}=-4N_3(2E_1+1)\\
&\ \ \ \ \ \ \ \ +(4E_1+2(N-N_1)(2E_1+1)-(2E_1+1)).
\end{aligned}
\end{equation}

When $\frac{\partial f(N_3)}{\partial N_3}=0$, we can solve $N_3$ as follows
\begin{equation}
\begin{aligned}
N_3=\frac{\frac{4}{2+1/E_1}+2(N-N_1)-1}{4}.
\end{aligned}
\end{equation}

$N_1\geq2$ can be obtained from Definition 8. At this time $E_1\geq1$,
therefore $\frac{4}{2+1/E_1}\leq2$.
Further we can derive $\frac{1+2(N-N_1)}{4}\geq\frac{\frac{4}{2+1/E_1}+2(N-N_1)-1}{4}$,
and let $N_3=\lfloor\frac{2(N-N_1)+1}{4}\rfloor$ because of $N_3\in\mathbb{N}_+$.


Next, we discuss the value of $N_1$. According to the above analysis,
it can be concluded that $N_2=\lceil\frac{2(N-N_1)-1}{4}\rceil$, $N_3=\lfloor\frac{2(N-N_1)+1}{4}\rfloor$, $M_1=\lceil\frac{N_1-1}{4}\rfloor$,
$M_2=N_1-2\lceil\frac{N_1-1}{4}\rfloor$.
And substituting these variables into DOF yields

\begin{equation}
\begin{aligned}
h(N_1)&\triangleq\text{DOF}=2(2(\lfloor\frac{2(N-N_1)+1}{4}\rfloor)+1)(2E_1\\
&\ \ \ \ \ \ \ \ +(\lceil\frac{2(N-N_1)-1}{4}\rceil)(2E_1+1))+1,\\
E_1&=-2(\lceil\frac{N_1-1}{4}\rfloor)^2+(N_1-1)\lceil\frac{N_1-1}{4}\rfloor+N_1-1.
\end{aligned}
\label{wang4}
\end{equation}

At this point, our goal is to obtain the maximum value of function $h(N_1)$ .
However, according to the analysis of (\ref{wang4}), since $h(N_1)$ is a cubic function of the independent variable $N_1$,
solving the maximum value is more complex for a cubic function with multiple extremum points \cite{Krentel1986}.
So we hope to find the maximum value of $h(N_1)$ through other methods.
Further analysis reveals that the value of $N_1$ ranges from 2 to $N$,
therefore constructing the Algorithm 1 to search the maximum DOF of FOGNA within the range from 2 to $N$
and the number of physical sensors of three subarrays for FOGNA.

\end{proof}

\begin{table}
\begin{center}
\label{tab1}
\begin{tabular}{ l }
\hline
\textbf{Algorithm 1: Search for the optimal array parameters } \\
\ \ \ \ \ \ \ \ \ \ \ \ \ \ \ \ \ \textbf{of FOGNA}  \\
\hline
\bf{Input}: $N$. \\
Initialization: $\text{DOF}^*=$ 0, $N_1^*=$ 0, \\
\ \ \ \ \ \ \ \ \ \ \ \ \ \ \ \ $N_2^*=$ 0, $N_3^*=$ 0.   \\
\ \ \ \ for \ \ $N_1=1$ to $N$.   \\
\ \ \ \ \ \ \ \ \ \ $N_2=\lceil\frac{2(N-N_1)-1}{4}\rceil$.\\
\ \ \ \ \ \ \ \ \ \ $N_3=\lfloor\frac{2(N-N_1)+1}{4}\rfloor$. \\
\ \ \ \ \ \ \ \ \ \ $M_1=\lceil \frac{N_1-1}{4} \rfloor$. \\
\ \ \ \ \ \ \ \ \ \ $M_2=N_1-2\lceil \frac{N_1-1}{4} \rfloor$. \\
\ \ \ \ \ \ \ \ \ \ $E_1=-2(\lceil\frac{N_1-1}{4}\rfloor)^2+(N_1-1)\lceil\frac{N_1-1}{4}\rfloor+N_1-1$.\\
\ \ \ \ \ \ \ \ \ \ $E_2=2E_1+N_2(2E_1+1)$.\\
\ \ \ \ \ \ \ \ \ \ $\text{DOF}=2(2N_3+1)(2E_1+N_2(2E_1+1))+1$ \\
\ \ \ \ \ \ \ \ \ \ if \ DOF $>$ $\text{DOF}^*$\\
\ \ \ \ \ \ \ \ \ \ \ \ \ \ $\text{DOF}^* =$ DOF.\\
\ \ \ \ \ \ \ \ \ \ \ \ \ $N_1^*=N_1$, $N_2^*=N_2$, $N_3^*=N_3$.\\
\ \ \ \ \ \ \ \ \ \ end\\
\ \ \ \ end \\
\textbf{Output}: $N_1^*$, $N_2^*$, $N_3^*$, $\text{DOF}^*$ $\rhd$ optimal array parameters.\\
\hline
\end{tabular}
\end{center}
\end{table}

\begin{corollary}
In the case of $N\equiv0 \ (mod \ 4)$, the upper bound DOF for FOGNA designed based on FOECA are
\begin{equation}
\mathcal{O}(\frac{N^4}{2}).
\end{equation}
\end{corollary}

\begin{proof}
See Appendix A.
\end{proof}

From \cite{Yang2023}, the upper bound DOF for different FODCAs and FOGNA are shown in Table I.
It can be seen that the proposed FOGNA greatly increases the DOF.

\begin{table}
\label{tab3}
\begin{center}
\caption{COMPARISON OF DOF FOR DIFFERENT ARRAYS}
\renewcommand{\arraystretch}{1.5} 
\begin{tabular}{ c c c c c }
\hline
\hline
\textbf{Array config.} & FL-NA & SE-FL-NA & FO-Fractal & FOGNA\\
\textbf{DOF}&$\mathcal{O}(\frac{N^4}{128})$ & $\mathcal{O}(\frac{N^4}{64})$ & $\mathcal{O}(\frac{N^4}{32})$ & $\mathcal{O}(\frac{N^4}{2})$\\
\hline
\hline
\end{tabular}
\end{center}
\end{table}

\subsection{Example of the Proposed FOGNA}

\begin{figure*}
\label{example}
 \center{\includegraphics[width=17cm]  {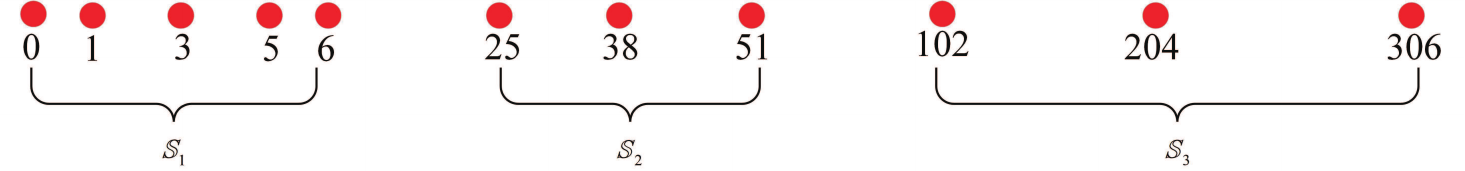}}
 \caption{\label{example} An example for FOGNA while $N=11$}
\end{figure*}

The array structure of FOGNA designed based on FOECA with $N=11$ physical sensors is shown in Fig. \ref{example}. The DOF of FOGNA are obtained as follows.

Firstly, by substituting $N=11$ into Algorithm 1,
the numbers of physical sensors for three subarrays of FOGNA are obtained as $N_1=5$, $N_2=3$ and $N_3=3$.
According to the array structure expression of CNA, the physical sensors of subarray 1 are located at positions given by the set $\mathbb{S}_1=\{ 0,1,3,5,6 \}\cdot d$.

Secondly, the SCA of subarray 1 with sensors located at positions given by the set $\{ 0:12 \}\cdot d$.
It means that subarray 1 can provide 13 consecutive lags based on SCA.
Furthermore, the union of $\mathbb{S}_1$ and $\mathbb{S}_2=\{ 25,38,51 \}\cdot d$ can constitute the linear array with sensors located at positions given by the set
$\mathbb{S}_2\cup \mathbb{S}_1=\{0,1,3,5,6,25,38,51 \}\cdot d$.
At this time, $p_{l_1},p_{l_2}\in\mathbb{S}_1$ and $p_{l_3}\in\mathbb{S}_2$,
the consecutive lags $\{-51:51\}\cdot d$ can be obtained from the cross-difference of $\mathbb{S}_1$ and $\mathbb{S}_2$.

Finally, to further enhance DOF, the physical sensor positions $\mathbb{S}_3=\{102,204,306 \}\cdot d$.
Therefore, we can get the FOGNA with 11 physical sensors located at positions given
by the set $\mathbb{S}=\mathbb{S}_1\cup\mathbb{S}_2\cup\mathbb{S}_3=\{ 0,1,3,5,6,25,38,51,102,204,306\}\cdot d$.
At the point, the consecutive lags for FOGNA are $\{-357:357\}\cdot d$,
and the DOF of FOGNA are $715$, which is higher than the DOF of $517$ for SD-FODC(NA) based on FODCA in \cite{Robin2017} when $N=11$.

\section{PERFORMANCE COMPARISON}
In this section, we provide numerical simulations to demonstrate the superior performance of
FOGNA in terms of DOF, coupling leakage, resolution and the RMSE versus the input SNR, snapshots and the number of sources.
Note that in DOA estimations, the spatial smoothing MUSIC algorithm \cite{Pal22011}, \cite{Liu2015}, \cite{Piya2012},
\cite{You2021} is usually used to estimate DOA.
Moreover, we assume that all incident sources have equal power and the number of sources is known.
To evaluate the results quantitatively, the root-mean-square error (RMSE) of the DOA estimation is defined as an
average over 1000 independent trials:

\begin{equation}
RMSE=\sqrt{\frac{1}{1000D}\sum_{j=1}^{1000}\sum_{i=1}^{D}(\hat{\theta}_i^{j}-\theta_i)^2},
\end{equation}
where $\hat{\theta}_i^{j}$ is the estimation of $\theta_i$ for the $j^{th}$ trial. Similar to \cite{LiuCL2016},
we focus on the DOF obtained by different arrays,
rather than the array aperture, to investigate the overall estimation performance.

\subsection{Comparison of the DOF for Different Arrays}
We compare the DOF of the proposed method with those of FL-NA \cite{Piya2012},
SE-FL-NA \cite{Shen2019}, FO-Fractal (NA) \cite{Yang2023} and SD-FODC (NA) \cite{Guo2024} for given the fixed number of physical sensors,
where FL-NA and SE-FL-NA adopt the  array structure for obtaining maximum DOF and the nested array is used as the  basic array of FO-Fractal and SD-FODC.
The comparing results are listed in Table II and the variations of DOF versus the number of sensors for five methods are shown in Fig. 3.
The results in Table II show that the DOF of FOGNA have significantly increased compared to other four methods under the different given number of physical sensors.
In addition, it can be seen that as the number of sensors $N$ increases, the DOF growth rate of the proposed FOGNA is faster than
those of the other four FODCAs in Fig. 3, and the advantages of the proposed array are more obvious with a large number of physical sensors.

\begin{table}
\begin{center}
\caption{COMPARISON OF DOF FOR DIFFERENT ARRAYS
 \\ BASED ON FOURTH-ORDER CUMULANT.}
\label{tab1}
\renewcommand{\arraystretch}{1.5} 
\begin{tabular}{ c  c  c  c }
\hline
\hline
\textbf{Array} & \textbf{Hole-Free} & \textbf{Number of} & \textbf{DOF}\\
\textbf{Structure}& \textbf{Co-Array}&\textbf{Sensors} & \\
\hline
FL-NA & Yes & $\sum\limits_{m=1}^{4}(N_m-1)+1$ &${L_{Fl}}^a$ \\
SE-FL-NA & Yes &$\sum\limits_{m=1}^{4}(N_m-1)+2$  &${L_{SE}}^b$ \\
FO-Fractal(NA) & Yes & $2N_1-1$ & ${2M^{2r}-1}^c$ \\
SD-FODC(NA) & Yes & $N_1+N_2$ & ${L_{SD}}^d$\\
FOGNA & Yes & $N_1+N_2+N_3$ & ${L_{new}}^e$\\
\hline
\hline
\textbf{Array} & \textbf{($N_1$,$N_2$) or} & \textbf{Number of} & \textbf{DOF}\\
\textbf{Structure} & \textbf{($N_1$,$N_2$,$N_3$)}& \textbf{Sensors}& \\
\hline
FL-NA & (3,3,3,3) & 9 & 217\\
SE-FL-NA & (3,3,3,2) & 9 & 253\\
FO-Fractal(NA) & (5,5) & 9 & 307\\
SD-FODC(NA) & (4,5) & 9 & 317\\
FOGNA & (5,2,2) & 9 & 381\\
\hline
FL-NA & (4,4,3,3) & 11 & 385\\
SE-FL-NA & (4,3,3,3) & 11 & 481\\
FO-Fractal(NA) & (6,6) & 11 & 553\\
SD-FODC(NA) & (6,5) & 11 & 597\\
FOGNA & (5,3,3) & 11 & 715\\
\hline
FL-NA & (6,6,5,5) & 19 & 2161\\
SE-FL-NA & (6,5,5,5)& 19 & 3121\\
FO-Fractal(NA) & (10,10) & 19 & 3541\\
SD-FODC(NA) & (10,9) & 19 & 3775\\
FOGNA & (9,5,5) & 19 & 4599\\
\hline
\hline
\end{tabular}
\end{center}
\footnotesize{$^a$ $L_{FL}=2(\Pi_{m=1}^4N_i+\Pi_{m=1}^3N_i)+1$}\par
\footnotesize{$^b$ $L_{SE}=N_3N_4(2N_1N_2-1)+(N_4-1)(N_1N_2-1)-1$}\par
\footnotesize{$^c$ $M^r=\frac{2N_G^2}{3}-\frac{2N_G}{3}+c_{t_0},N_G=\frac{N+1}{2}$}\par
\footnotesize{$^d$ $L_{SD}=(4d_m + 2)d_n + ((2d_m + 1)\frac{\mu_2+1}{2}+\frac{\mu_1-1}{2})$}\par
\footnotesize{$^e$ $L_{new}=2[(-2N_3^2-N_3)(2E_1+1)+(2E_1+(N-N_1)(2E_1+1))(2N_3+1)]+1$}\\
\end{table}

\begin{figure}[H]
 \center{\includegraphics[width=6cm]  {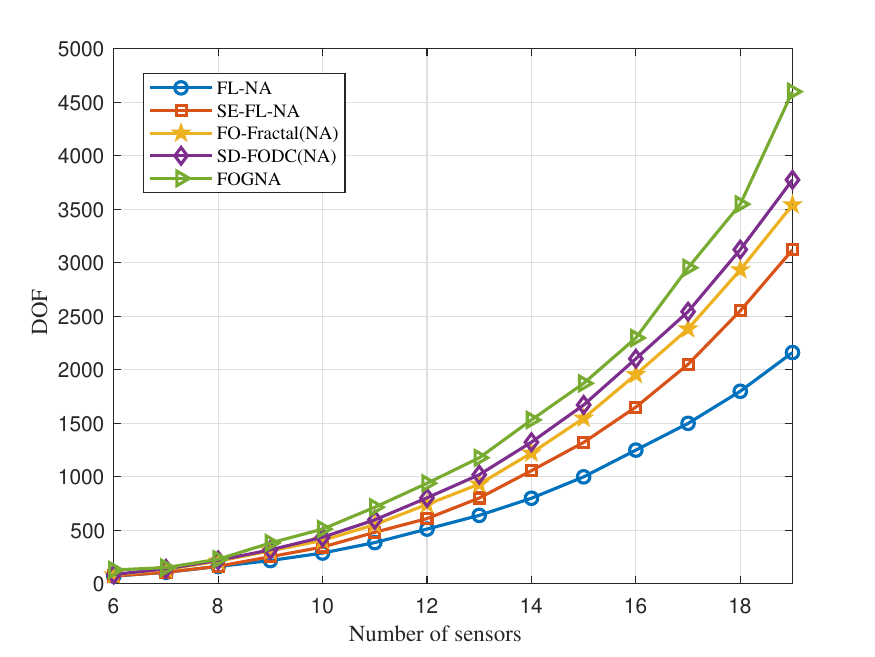}}
 \caption{\label{1} DOF of different arrays}
\end{figure}

\subsection{Mutual Coupling Matrices}
The mutual coupling performance of the proposed co-array is compared with those of state-of-the-art co-arrays, FL-NA, SE-FL-NA, FO-Fractal (NA) and SD-FODC (NA), in term of coupling leakages in this section. As mentioned above, both FO-Fractal and SD-FODC are also using nested arrays as their basic arrays.
Specifically, the mutual coupling model (\ref{wang5}) is characterized by
$c_1=0.3e^{j\pi/3}$, $B=100$ and $c_l=c_1e^{-j(l-1)\pi/8}/l$, for $2\leq l \leq B$.

Firstly, the coupling leakages are calculated by (\ref{w24}) for different configurations with the number of physical sensors
as 9, 10, 11, 9, 21 and 23, respectively.
The results of the coupling leakages of different structure for five co-arrays are listed in Table III.
It can be seen that the coupling leakage $L$ of FOGNA is lower than those of other four FODCAs when the number of sensors $N > 10$.
However, the performance of the coupling leakage for other FODCAs is better than that of
FOGNA when the number of sensors $N = 9$,
where  the reason is that there exist more physical senors with unit inter-spacing in subarray 1 of FOGNA
than those of other FODCAs, resulting in increased coupling leakage of FOGNA.
The same situation can be seen in Table III when $N=10$.
Therefore, the proposed array structure outperforms the other four FODCAs under the case of $N>10$ in terms of coupling leakage.

\begin{table}
\label{tab2}
\begin{center}
\caption{A SUMMARY OF MUTUAL COUPLING LEAKAGE FOR FIVE ARRAY STRUCTURES}
\renewcommand{\arraystretch}{1.5} 
\begin{tabular}{ c  c  c  c  c  c}
\hline
\hline
\textbf{Array } & \textbf{ FL-NA } & \textbf{SE-FL} & \textbf{FO-Fractal} & \textbf{SD-FODC} & \textbf{FOGNA}\\
\textbf{config.} &  & \textbf{-NA} & \textbf{(NA)} & \textbf{(NA)} & \\
\hline
9 sensors    & (3,3,3,3) & (3,3,3,2)  & (5,5)      &(4,5)      & (4,2,3)\\
$L$          & 0.2263    & 0.2257     & 0.2247     & 0.2187    &0.2347  \\
\hline
10 sensors    & (4,3,3,3) & (3,3,3,3)  & (5,6)      &(5,5)      & (4,3,3)\\
$L$          & 0.2563    & 0.2147     & 0.2137     & 0.2139    &0.2236  \\
\hline
11 sensors   & (4,4,3,3) & (4,3,3,3)  & (6,6)      & (6,5)     & (5,3,3)\\
$L$          & 0.2477    & 0.2449     & 0.2444     & 0.2446    &\textbf{0.2137}  \\
\hline
19 sensors   & (6,6,5,5) & (6,5,5,5)  &  (10,10)   &(10,9)     & (9,5,5)\\
$L$          & 0.2460    & 0.2452     & 0.2451     & 0.2452    &\textbf{0.2018}  \\
\hline
21 sensors   & (6,6,6,6) & (6,6,6,5)  &  (11,11)   &(11,10)    & (9,6,6)\\
$L$          & 0.2347    & 0.2346     & 0.2346     & 0.2346    & \textbf{0.2139}  \\
\hline
23 sensors   & (7,7,6,6) & (7,6,6,6)  &  (12,12)   &(12,11)    & (12,5,6)\\
$L$          & 0.2464    & 0.2459     & 0.2459     & 0.2459    &\textbf{0.2077}  \\
\hline
\hline
\end{tabular}
\end{center}
\end{table}


Secondly, the visualizations of mutual coupling matrices for different arrays are shown in Fig. 4,
with the number of sensors set as 23,
where the darker blue color corresponds to the smaller value of the non-diagonal elements
in mutual coupling matrices and yellow color represents the values of diagonal elements of matrices.
It is observed in Fig. 4 that the higher mutual coupling of FL-NA, SE-FL-NA, FO-Fractal (NA) and SD-FODC (NA) occurs in the dense subarray part,
i.e., the nested array, whereas higher mutual coupling of FOGNA is concentrated at both ends of the
subarray 1, where the senor spacing is denser.

\begin{figure*}
  \centering
  \subfigure[]{
    \label{fig:subfig:onefunction}
    \includegraphics[scale=0.22]{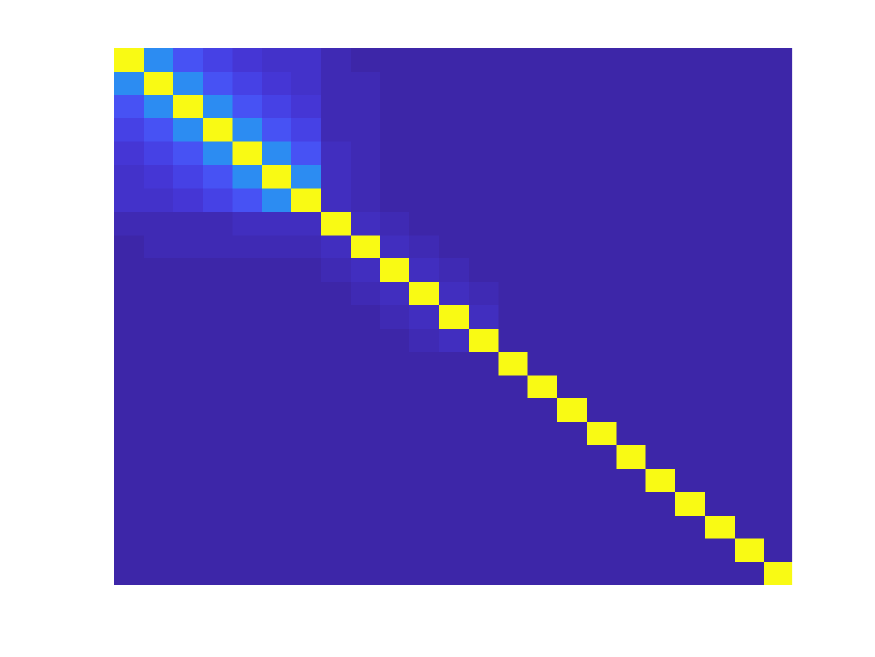}}
  \hspace{0in} 
  \subfigure[]{
    \label{fig:subfig:threefunction}
    \includegraphics[scale=0.22]{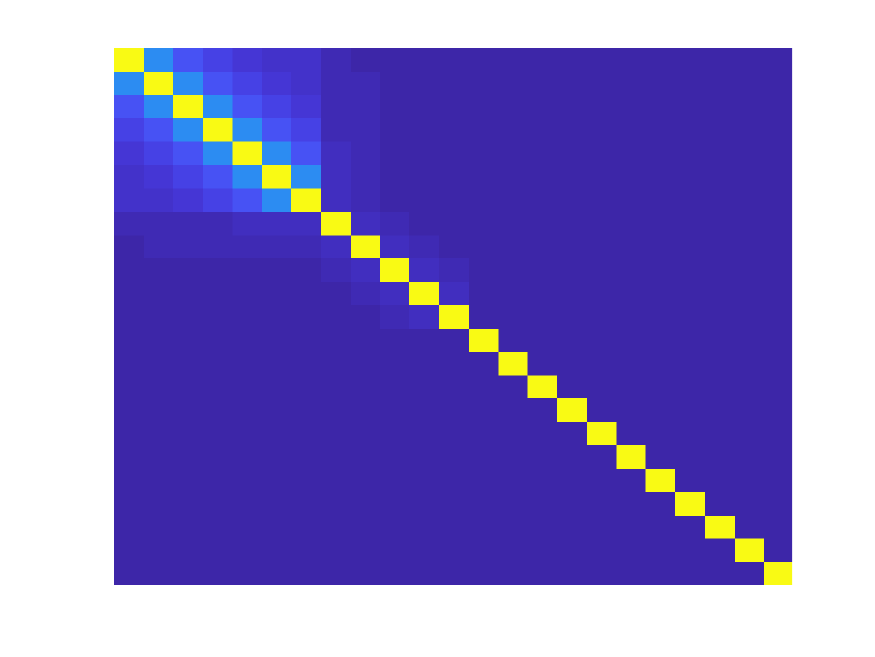}}
  \hspace{0in} 
  \subfigure[]{
    \label{fig:subfig:threefunction}
    \includegraphics[scale=0.22]{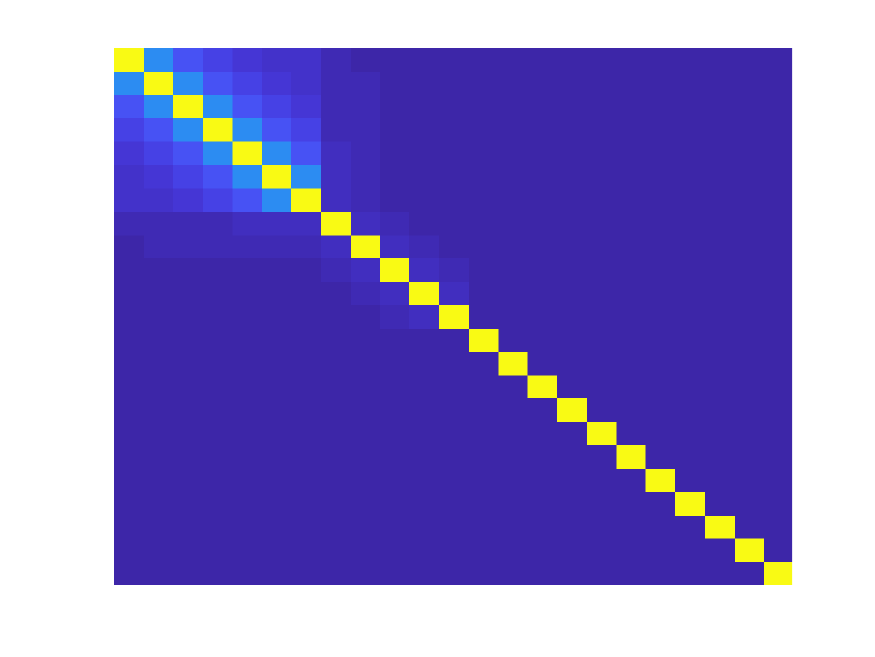}}
    \hspace{0in} 
  \subfigure[]{
    \label{fig:subfig:threefunction}
    \includegraphics[scale=0.22]{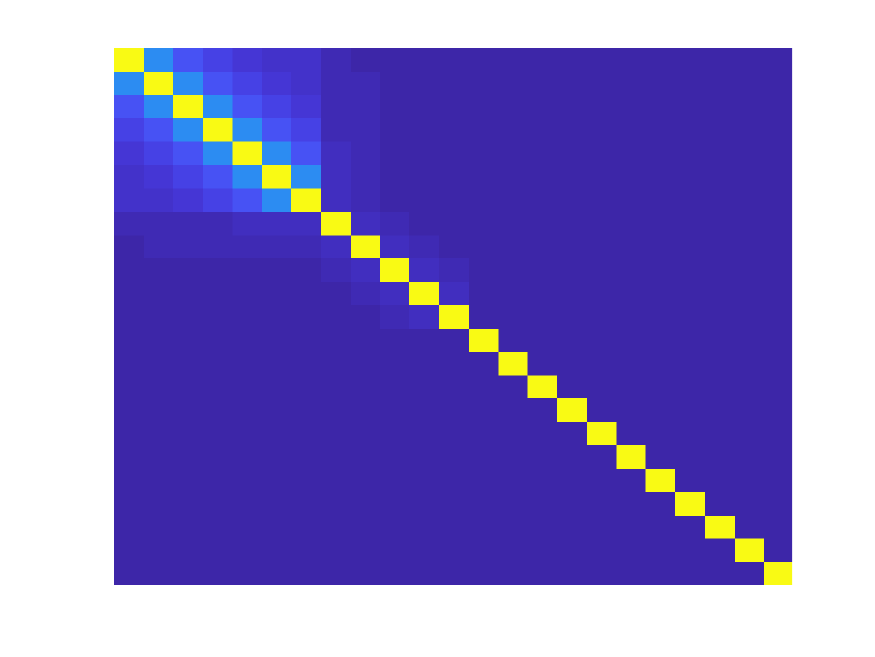}}
    \hspace{0in} 
  \subfigure[]{
    \label{fig:subfig:threefunction}
    \includegraphics[scale=0.22]{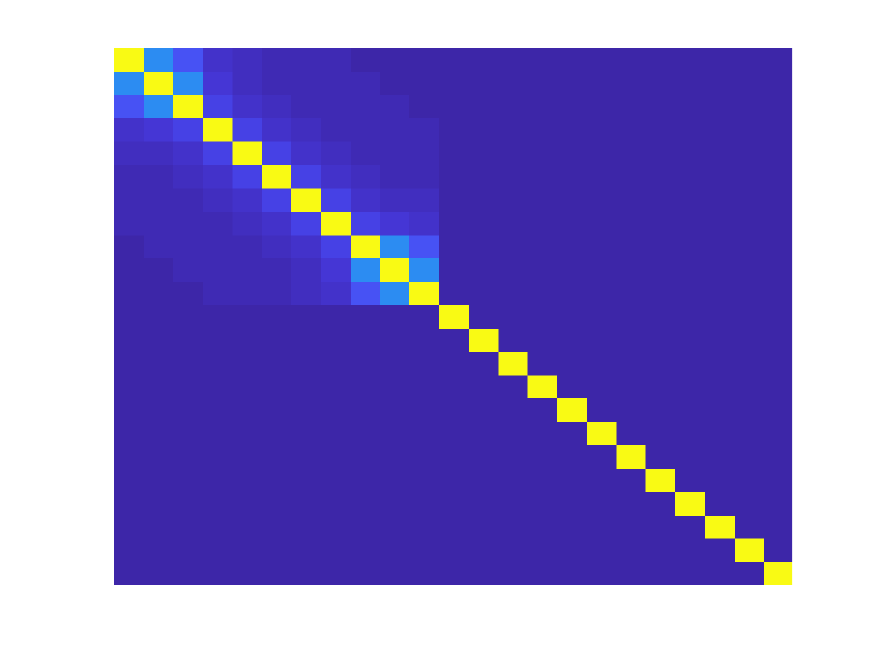}}
  \caption{The magnitudes of the mutual coupling matrices of five arrays with 23-sensors. (a) FL-NA. (b) SE-FL-NA. (c) FO-Fractal(NA). (d) SD-FODC(NA).
(e) FOGNA.}
\end{figure*}

\subsection{Resolution of Different Array Structures}
\begin{figure*}
  \centering
  \subfigure[]{
    \label{fig:subfig:onefunction}
    \includegraphics[scale=0.23]{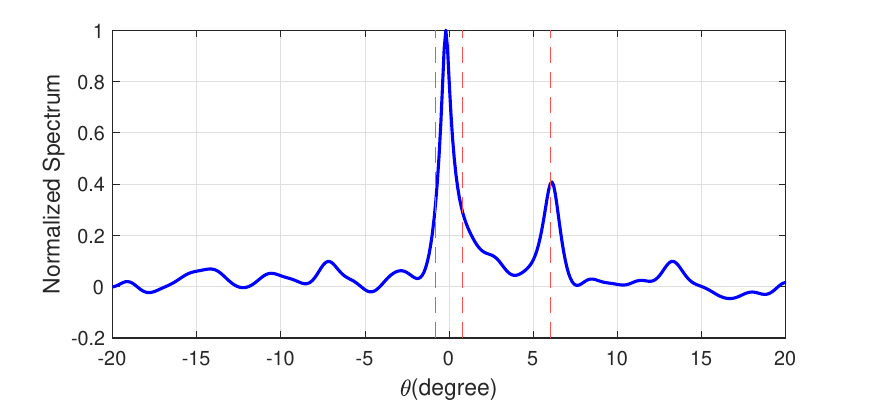}}
  \hspace{0in} 
  \subfigure[]{
    \label{fig:subfig:threefunction}
    \includegraphics[scale=0.23]{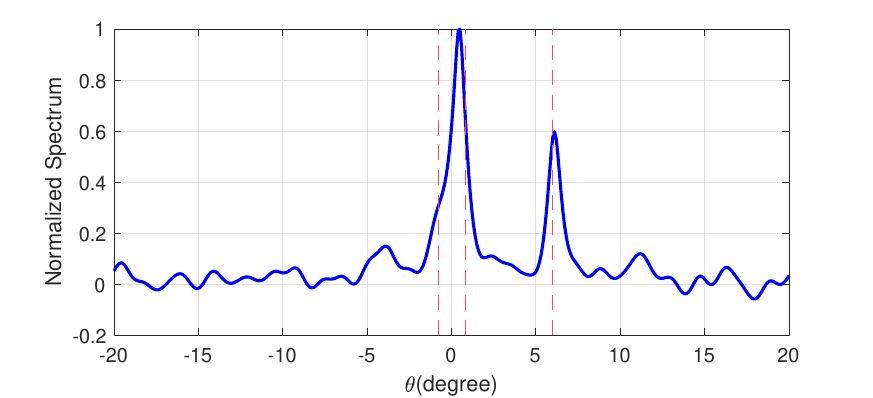}}
  \subfigure[]{
    \label{fig:subfig:threefunction}
    \includegraphics[scale=0.23]{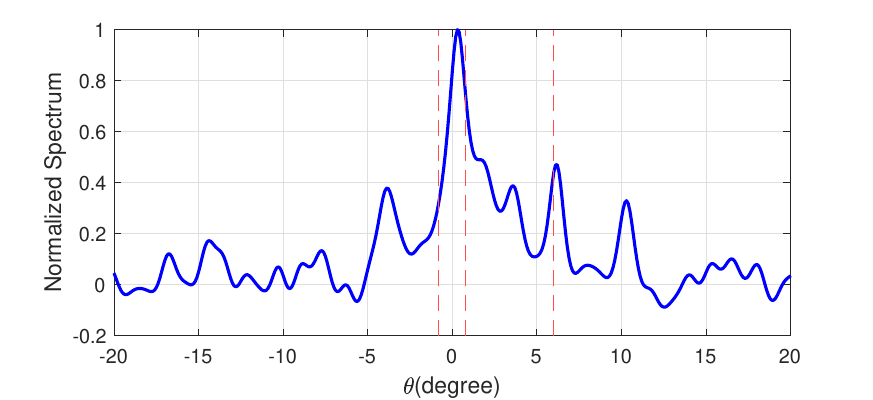}}
  \subfigure[]{
    \label{fig:subfig:threefunction}
    \includegraphics[scale=0.23]{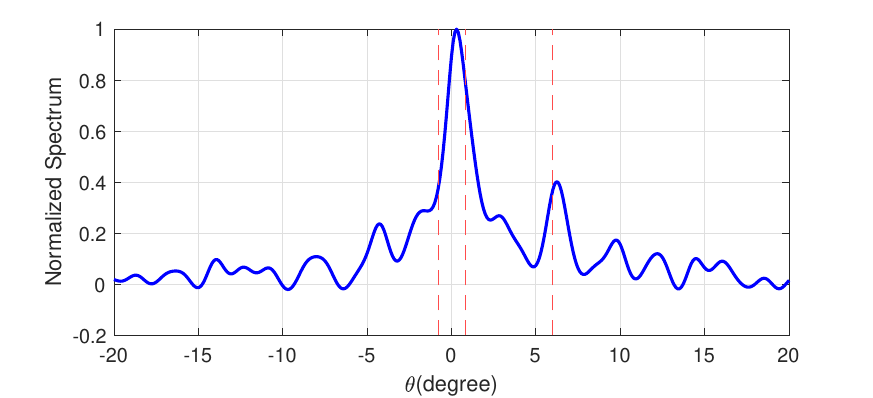}}
  \subfigure[]{
    \label{fig:subfig:threefunction}
    \includegraphics[scale=0.23]{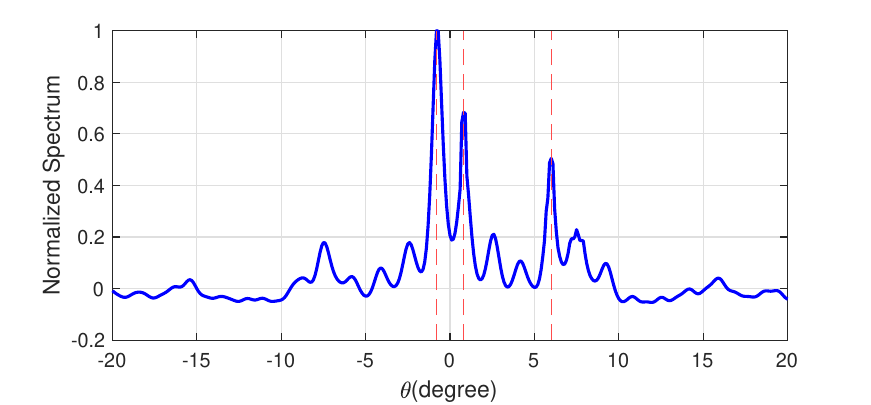}}
  \caption{The special case of DOA estimation result for five arrays with 7-sensors when sources are located at $-0.8^0$, $0.8^0$ and $6^0$.
$SNR = 0$ dB and $K = 10000$. (a) FL-NA. (b) SE-FL-NA. (c) FO-Fractal(NA). (d) SD-FODC(NA). (e) FOGNA.}
\end{figure*}



The resolution is an important metric of the performance for DOA estimation, which is compared among FL-NA, SE-FL-NA, FO-Fractal (NA), SD-FODC (NA)
and FOGNA in the simulation. 7 physical sensors are used to construct five co-array,
while the angle of one source fixed at $-0.8^{\circ}$, and the angle of other sources vary from 0 to 10.
The SNR and snapshots are set as 0 dB and 10000, respectively.
The angles of 3 uncorrelated sources are $-0.8^{\circ}$, $0.8^{\circ}$ and $6^{\circ}$ respectively in Fig 5.
It can be seen that FOGNA can distinguish two sources with the difference of incidence angles reduced to $1.6^{\circ}$, while others can not.

\subsection{DOA Estimation Without Mutual Coupling}
We compare the DOA estimation performance versus input SNR, snapshots and the number of sources for FOGNA without mutual coupling to
those of other four FODCAs in this part,
where 9 physical sensors are used to construct five co-array.

Firstly, there are 40 uncorrelated sources uniformly located at $-60^{\circ}$ to $60^{\circ}$,
and the SNR and snapshots are set as 0 dB and 10000, respectively.
The DOA estimation results are shown in Fig. 6, where only the FL-NA is incapable of resolving all 40 sources.
Furthermore, the FOGNA exhibits lower valley near both ends than those of FL-NA, which can improve the performance of DOA estimation.

\begin{figure*}
  \centering
  \subfigure[]{
    \label{fig:subfig:onefunction}
    \includegraphics[scale=0.15]{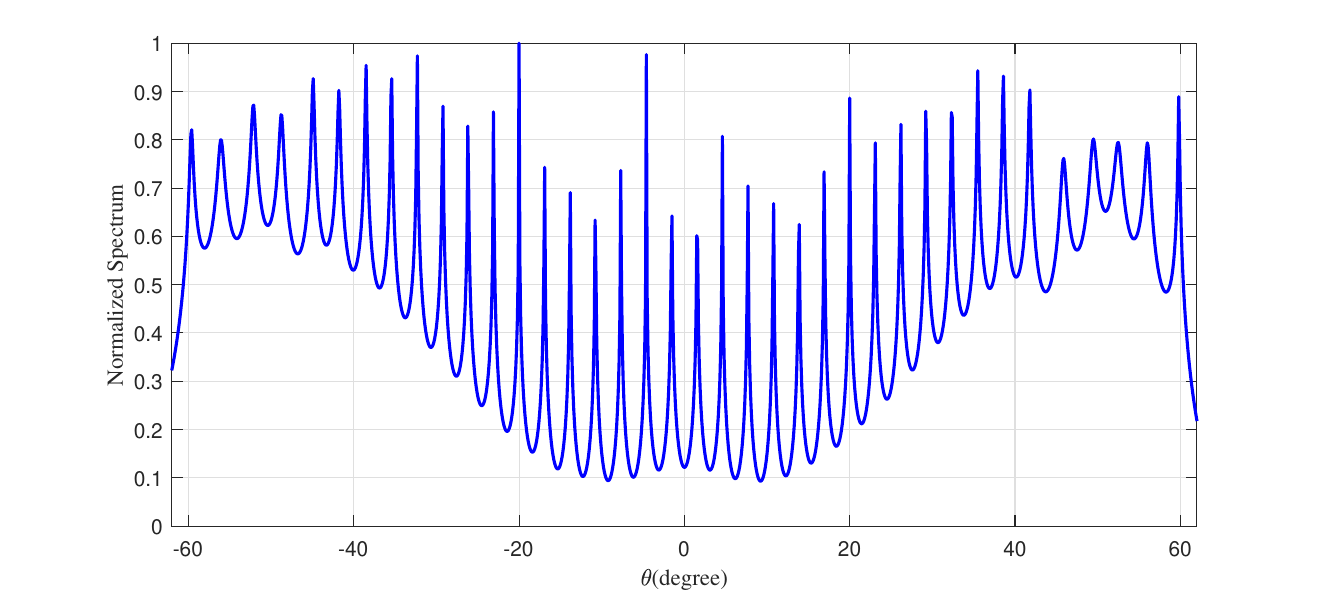}}
  \hspace{0in} 
  \subfigure[]{
    \label{fig:subfig:threefunction}
    \includegraphics[scale=0.15]{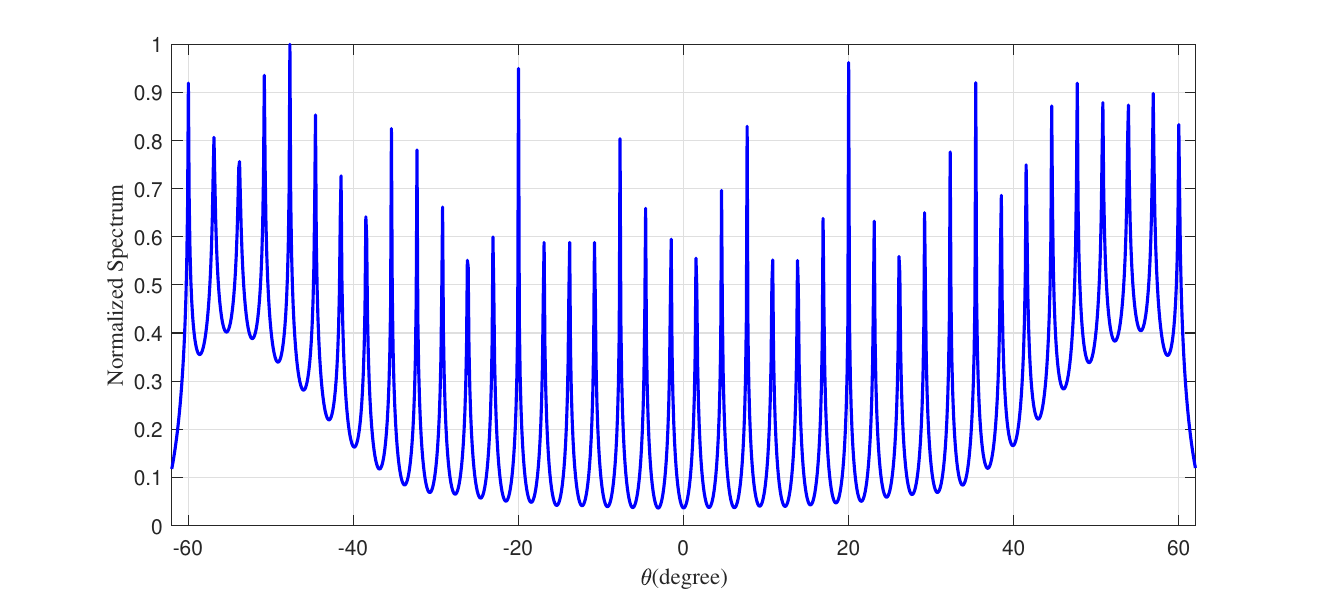}}
  \subfigure[]{
    \label{fig:subfig:threefunction}
    \includegraphics[scale=0.15]{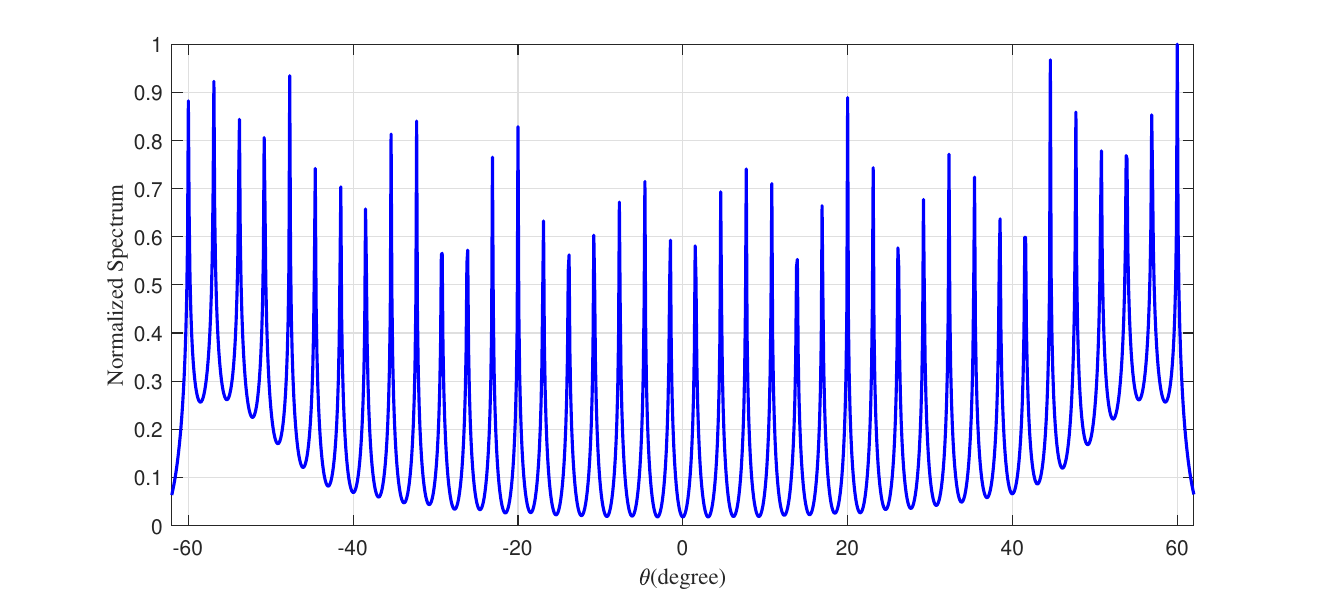}}
  \subfigure[]{
    \label{fig:subfig:threefunction}
    \includegraphics[scale=0.15]{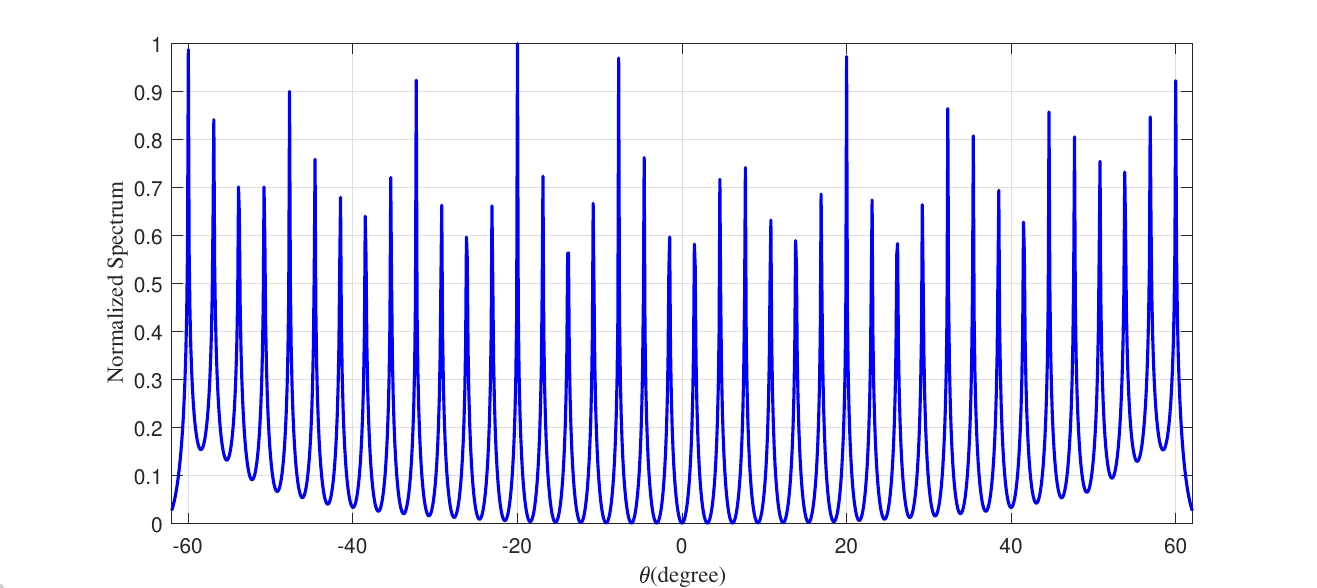}}
  \subfigure[]{
    \label{fig:subfig:threefunction}
    \includegraphics[scale=0.15]{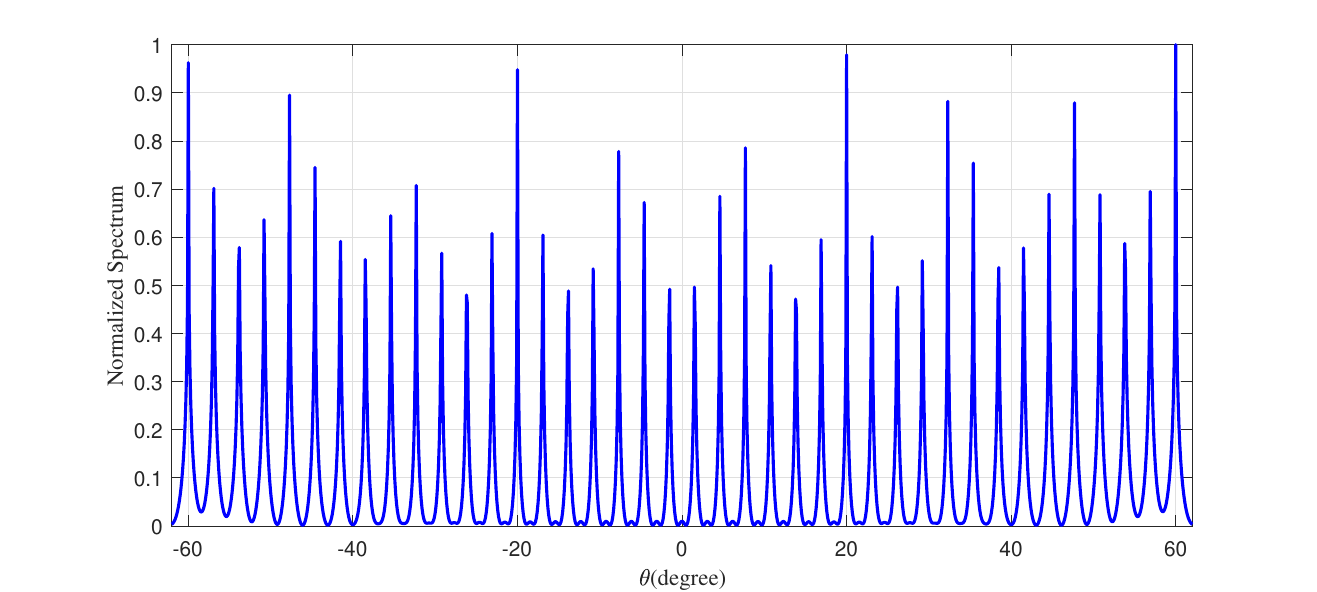}}
  \caption{DOA estimation result for five arrays with 9-sensors when 40 sources are uniformly located at $-60^0$ to $60^0$.
$SNR = 0$ dB and $K = 10000$. (a) FL-NA. (b) SE-FL-NA. (c) FO-Fractal(NA). (d) SD-FODC(NA). (e) FOGNA.}
\end{figure*}

Secondly, the RMSE versus the input SNR, snapshots and the number of sources are studied in the following numerical simulations.
In the first numerical simulation, there are 12 uncorrelated sources uniformly located at $-60^{\circ}$ to $60^{\circ}$
and the snapshots setting as 14000. The SNR ranges from -7dB to 8dB with an interval of 3dB.
The results of RMSE versus SNR for different arrays are shown in Fig. 7(a),
where it can be seen that as the SNR increases, the RMSEs of all arrays decrease,
however the RMSE of FOGNA remaining the lowest.
The second numerical simulation studies the DOA estimation performance with respect to the snapshots changing from 10000 to 16000
and the SNR setting as 5dB.
The results of RMSE versus snapshots are shown in Fig. 7(b), and a similar conclusion can be obtained that
the RMSE of FOGNA is significantly lower than those of other four FODCAs.
In the third numerical simulation, the number of sources change from 12 to 20.
The results of RMSE versus the number of sources are shown in Fig. 7(c),
where it can be seen that as the number of sensors increases, the RMSEs of FL-NA and SE-FL-NA increase steeply
than those of FO-Fractal (NA), SD-FODC (NA) and FOGNA.
Notably, the RMSE of FOGNA remains the smallest compared to other four FODCAs, indicating its superior performance.
\begin{figure}
  \centering
  \subfigure[ RMSE versus SNR]{
    \label{fig:subfig:onefunction}
    \includegraphics[scale=0.18]{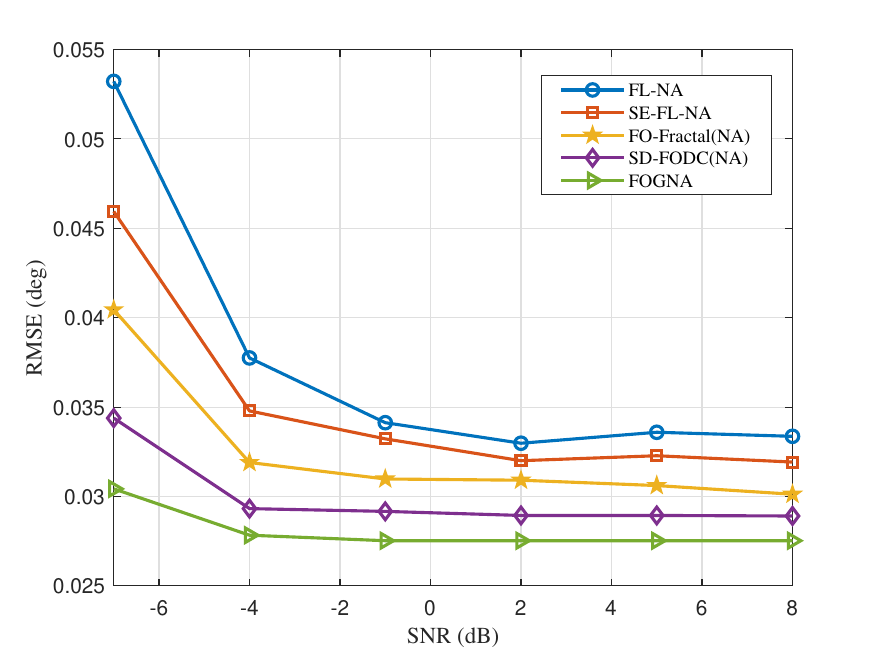}}
   \hspace{0in} 
  \subfigure[RMSE versus Snapshots]{
    \label{fig:subfig:threefunction}
    \includegraphics[scale=0.18]{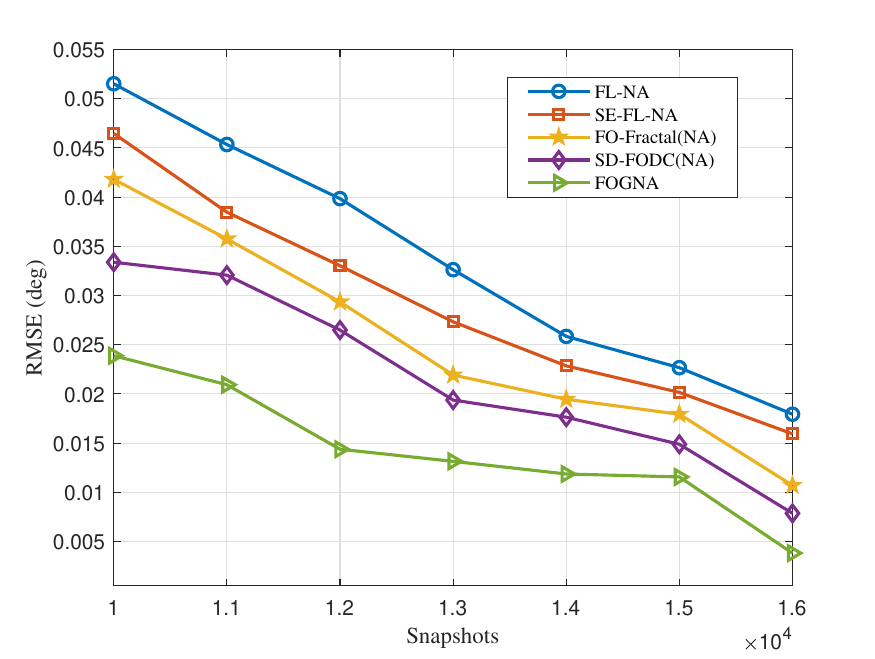}}
    \hspace{0in} 
  \subfigure[RMSE versus The number of sources]{
    \label{fig:subfig:threefunction}
    \includegraphics[scale=0.18]{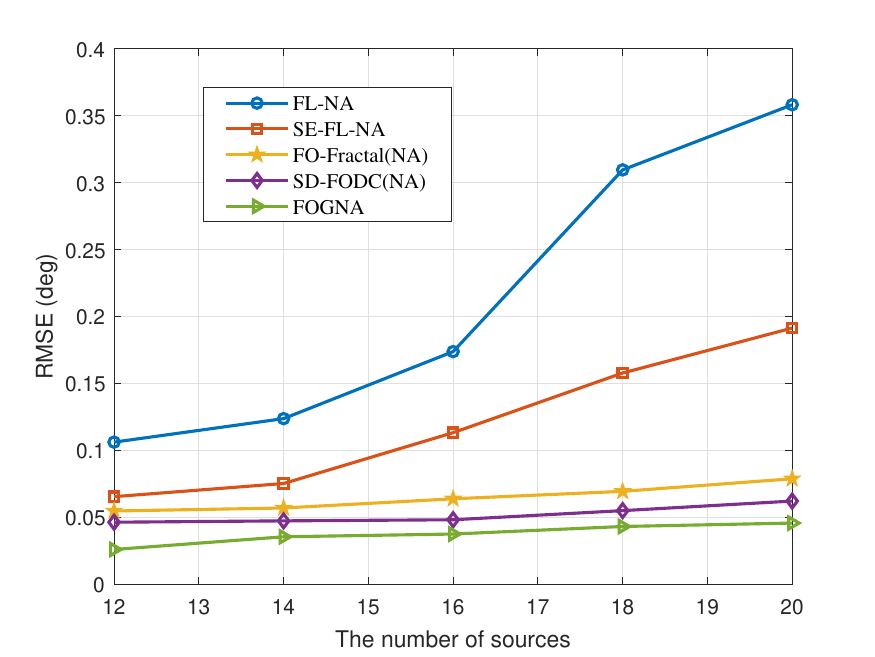}}
  \caption{DOA estimation performance without mutual coupling based on fourth-order cumulants}
\end{figure}

\subsection{DOA Estimation With Mutual Coupling}
We compare the DOA estimation performance versus input SNR, snapshots and the number of sources for FOGNA with mutual coupling to
those of other four FODCAs in this part,
where 11 physical sensors are used to construct five co-array.

Firstly, there are 50 uncorrelated sources uniformly located at $-60^{\circ}$ to $60^{\circ}$,
and the SNR and snapshots are set as 0 dB and 12000, respectively.
The DOA estimation results are shown in Fig. 8, where only the FOGNA is capable of resolving all 50 sources successfully.
It implies that the FOGNA is more effective than other four FODCAs against mutual coupling.

\begin{figure*}
  \centering
  \subfigure[]{
    \label{fig:subfig:onefunction}
    \includegraphics[scale=0.14]{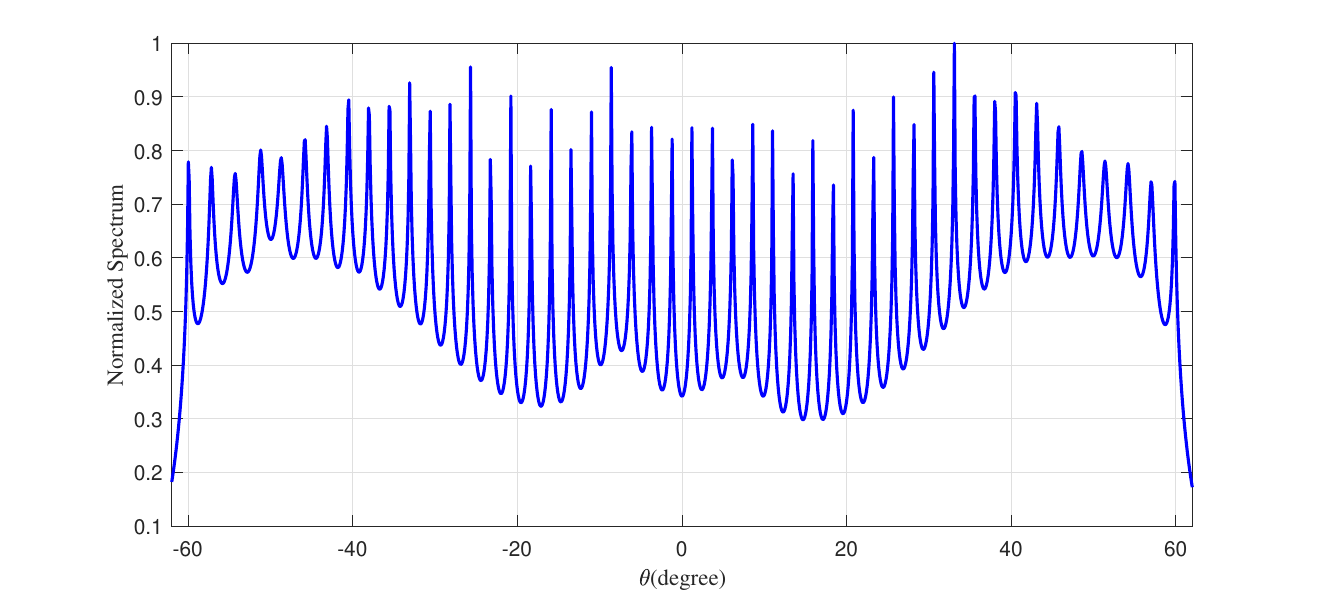}}
  \hspace{0in} 
  \subfigure[]{
    \label{fig:subfig:threefunction}
    \includegraphics[scale=0.14]{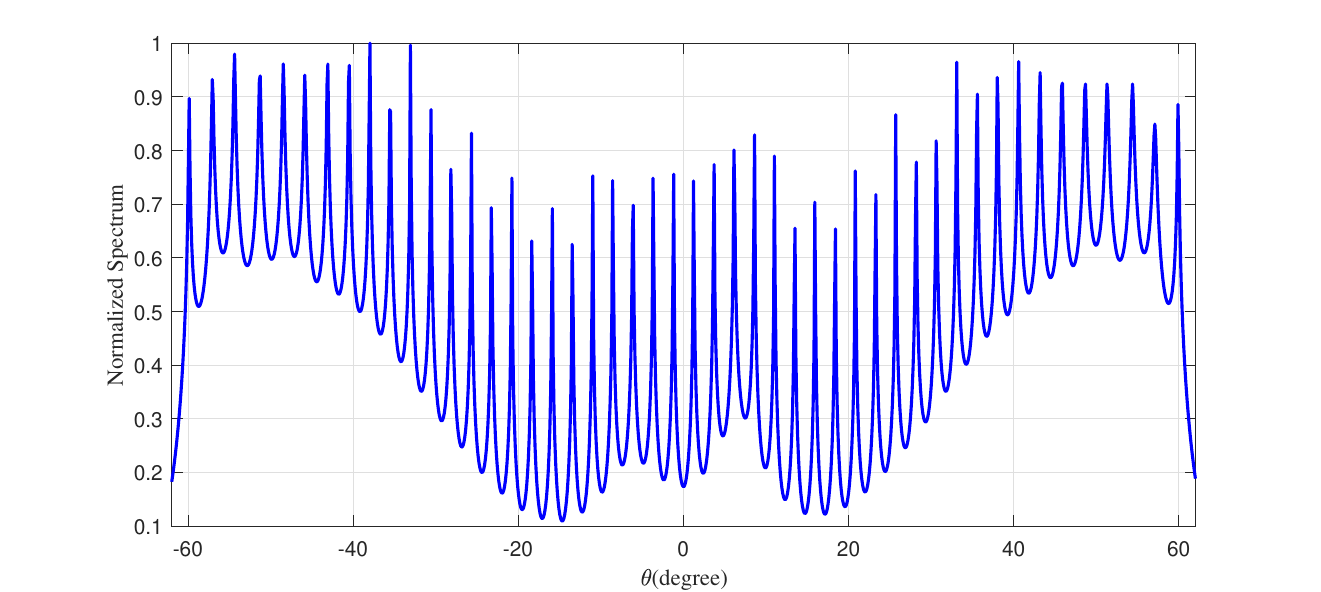}}
  \subfigure[]{
    \label{fig:subfig:threefunction}
    \includegraphics[scale=0.14]{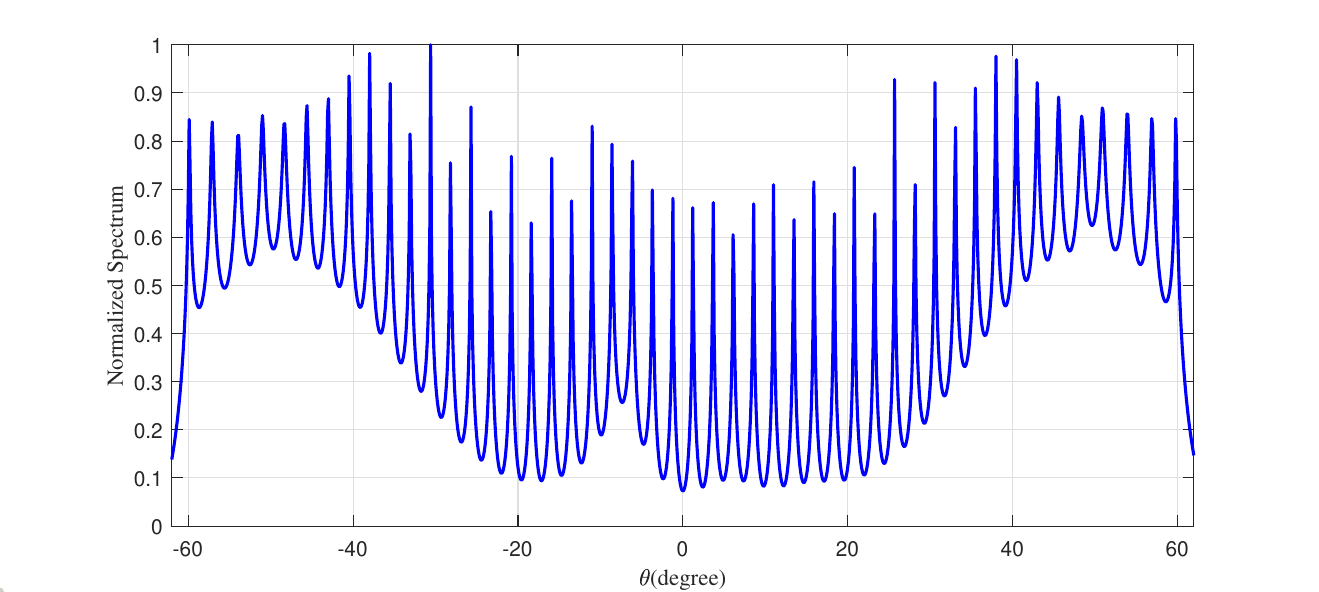}}
  \subfigure[]{
    \label{fig:subfig:threefunction}
    \includegraphics[scale=0.14]{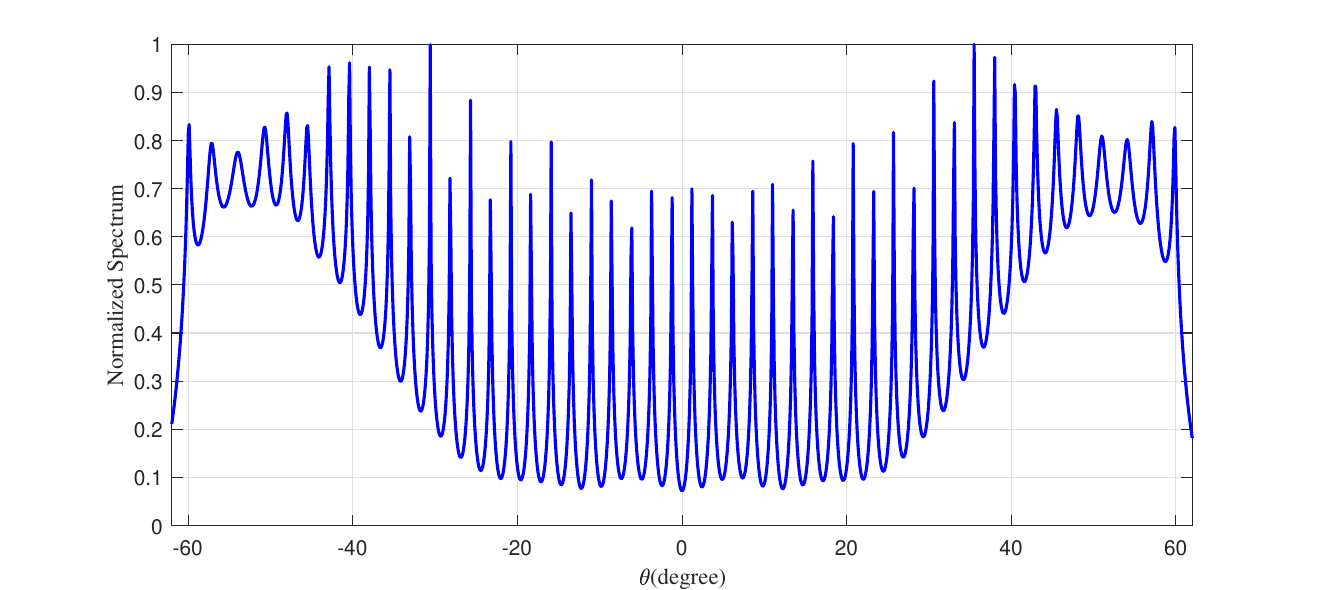}}
  \subfigure[]{
    \label{fig:subfig:threefunction}
    \includegraphics[scale=0.14]{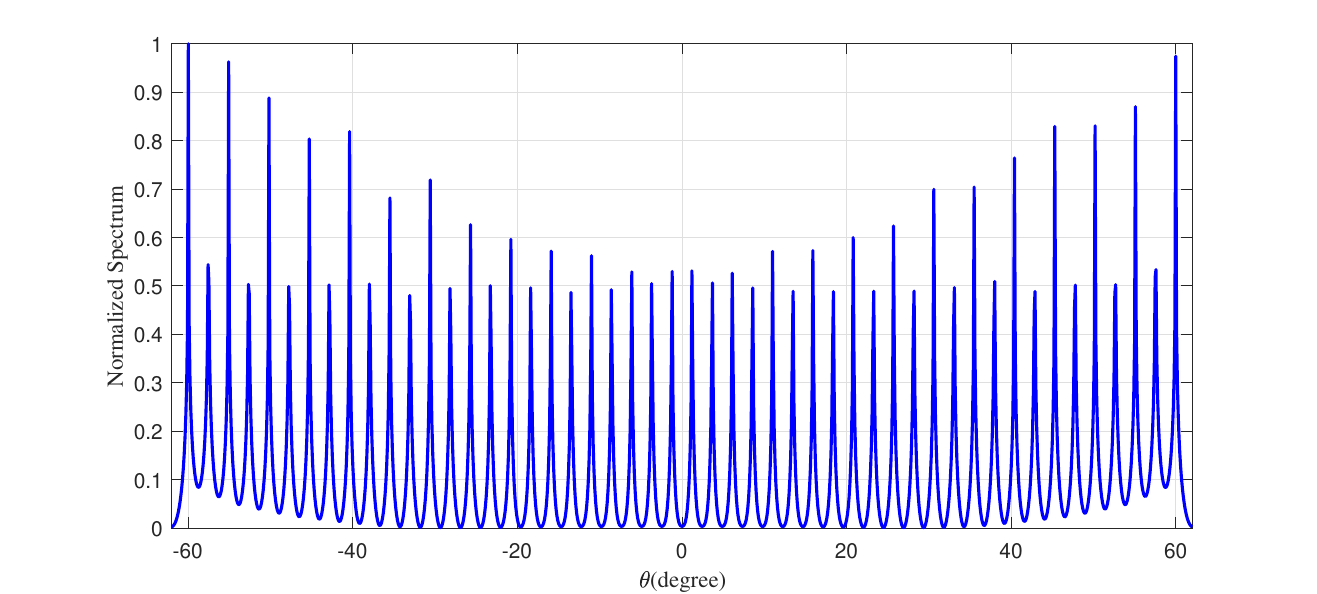}}
  \caption{ DOA estimation result for five arrays with 11-sensors when 50 sources are uniformly located at $-60^0$ to $60^0$.
$SNR = 0$ dB and $K = 12000$. (a) FL-NA. (b) SE-FL-NA. (c) FO-Fractal(NA). (d) SD-FODC(NA). (e) FOGNA.}
\end{figure*}

Secondly, the next numerical simulations focus on the RMSE versus the input SNR, snapshots and the number of sources.
In the first numerical simulation, there are 18 uncorrelated sources uniformly located at $-60^{\circ}$ to $60^{\circ}$,
and the snapshots are set as 15000. The SNR ranges from -7dB to 8dB with an interval of 3dB.
The results of RMSE versus SNR for different arrays are shown in Fig. 9(a),
where the FOGNA yields the lowest RMSE across the entire signal-to-noise ratio range.
It implies that FOGNA outperforms compared to the other four FODCAs in terms of mutual coupling effect.
The second numerical simulation studies the DOA estimation performance with respect to the snapshots changing from 10000 to 16000
and the SNR setting as 5dB.
It is shown that as the snapshots increase, the RMSEs of all arrays decrease,
while the RMSE of FOGNA is always smaller than those of other four FODCAs.
In the third numerical simulation, the number of sources change from 12 to 20.
The results of RMSE versus the number of sources are shown in Fig. 7(c),
which reveals that as the number of sources increase, the RMSEs of all arrays rose accordingly.
Notably, the RMSE of FOGNA increases slowly and remains lower than those of four FODCAs.

\begin{figure}
  \centering
  \subfigure[ RMSE versus SNR]{
    \label{fig:subfig:onefunction}
    \includegraphics[scale=0.18]{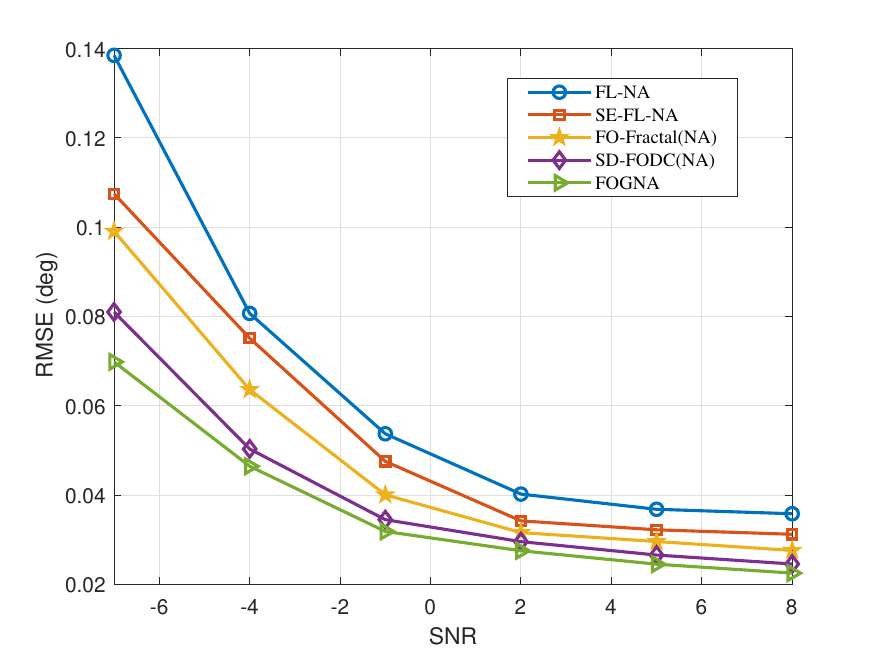}}
  \hspace{0in} 
  \subfigure[RMSE versus Snapshots]{
    \label{fig:subfig:threefunction}
    \includegraphics[scale=0.18]{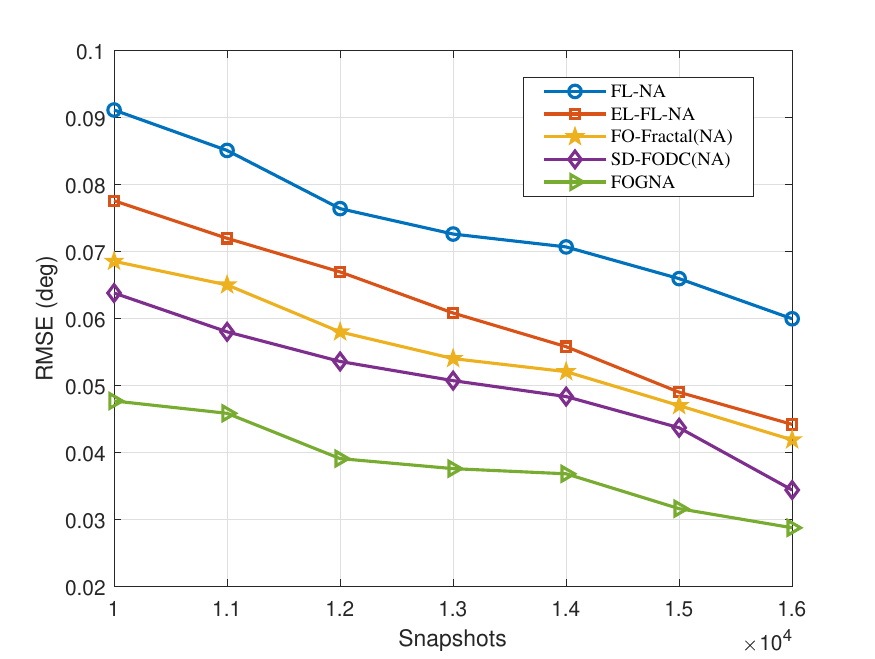}}
  \hspace{0in} 
  \subfigure[RMSE versus The number of sources]{
    \label{fig:subfig:threefunction}
    \includegraphics[scale=0.18]{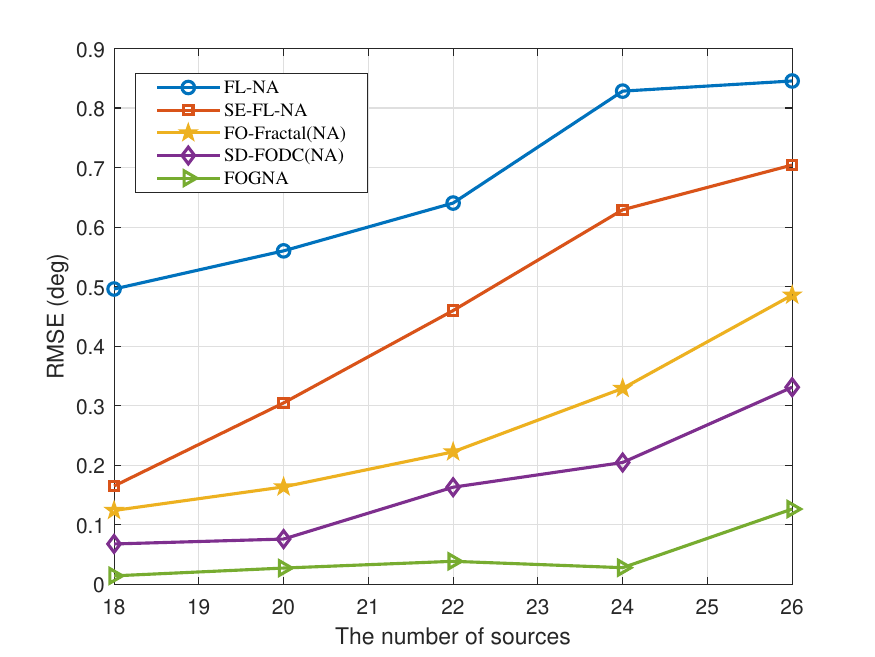}}
  \caption{DOA estimation performance with mutual coupling based on fourth-order cumulants}
\end{figure}

\section{Conclusion}
This paper exploits high-order statistics, namely fourth-order cumulants,
to devise the FOECA. Further a novel sparse linear array, namely FOGNA, is proposed based on FOECA,
which can enhance the DOF and improve the DOA estimation performance.
FOGNA consists of three subarrays with a given number of sensors.
Subarray 1 is capable to produce a hole-free sum co-array with more consecutive lags,
and both subarray 2 and 3 can produce a hole-free difference co-array with more consecutive lags.
Based on these criterions, CNA is selected as subarray 1.
Subarray 2 and 3 are a sparse linear array with a large inter-spacing between two physical sensors.
This arrangement yields closed-form expressions for the sensor positions of the proposed FOGNA by Algorithm 1.
Therefore, based on this design, the proposed FOGNA offers significantly larger DOF than those of other existing FODCAs.
The enhanced DOF increase the number of resolvable sources with the given number of physical sensors.
Meanwhile, compared to existing FODCAs, FOGNA exhibits lower mutual coupling and higher resolution.
In addition, the proposed FOGNA needs more snapshots than the sparse linear array based on second-order statistics,
as estimators based on high-order statistics entail a relatively higher variance.
Designing a suitable DOA estimation algorithm \cite{Kane2024}
to effectively exploit FOECA is a potential future work of this paper.

\section*{Acknowledgment}
This work was supported by the National Natural Science Foundation of China (Grant No. 62371400).\\

{\appendices
\section{Proof of the Corollary 1}
From (\ref{wang4}), we can get the DOF and $E_1$ of FOGNA. Due to $0\leq N_1 \leq N$, they are enlarged as follows,
\begin{equation}
\begin{aligned}
\text{DOFs} \leq &2 \{ 2(\lfloor \frac{2N+1}{4} \rfloor +1)+(2E_1+\lceil \frac{2N-1}{4} \rceil)(2E_1+1)\}+1,\\
E_1\leq&(N-1)\lceil\frac{N-1}{4}\rfloor+N-1,
\end{aligned}
\end{equation}
substituting $E_1$ into DOF, we can get
\begin{equation}
\begin{aligned}
\text{DOF} &\leq 2 \{ 2(\lfloor \frac{2N+1}{4} \rfloor +1)+(2((N-1)\lceil\frac{N-1}{4}\rfloor+N-1)\\
&+\lceil \frac{2N-1}{4} \rceil)(2((N-1)\lceil\frac{N-1}{4}\rfloor+N-1)+1)\}+1,\\
&\leq 2 \{ 2(\lfloor \frac{2N+1}{4} \rfloor +1)+(2(N\lceil\frac{N}{4}\rfloor+N)+\lceil \frac{2N}{4} \rceil)\\
&(2(N\lceil\frac{N}{4}\rfloor+N)+1)\}+1,\\
\end{aligned}
\end{equation}
when $N\equiv0\ (mod\ 4)$, the upper bound DOF are
\begin{equation}
\begin{aligned}
\mathcal{O}(\text{DOF}) &= \mathcal{O}(2\cdot 2\cdot \frac{N^2}{4}\cdot 2N\cdot \frac{N}{4})=\mathcal{O}(\frac{N^4}{2}).
\end{aligned}
\end{equation}

}


\begin{thebibliography}{[99]}

\bibitem{Xiao2017notes}
H. S. Xiao and G. Q. Xiao, ``Notes on CRT-based robust frequency estimation,''
\textit{Signal processing.}, vol. 133, pp. 13-17, 2017.

\bibitem{Xiao2018robustness}
H. S. Xiao, Y. F. Huang, Y. Ye and G. Q. Xiao, ``Robustness in chinese remainder theorem for multiple numbers and remainder coding,''
\textit{IEEE Transactions on Signal Processing.}, vol. 66, no. 16, pp. 4347-4361, 2018.


\bibitem{Xiao2016symmetric}
H. S. Xiao, C. Cremers and H. K. Garg, ``Symmetric polynomial $\&$ CRT based algorithms for multiple frequency determination from undersampled waveforms,''
\textit{2016 IEEE Global Conference on Signal and Information Processing (GlobalSIP).}, pp. 202-206, 2016.

\bibitem{Xiao2021wrapped}
H. S. Xiao, N. Du, Z. K. Wang and G. Q. Xiao, ``Wrapped ambiguity Gaussian mixed model with applications in sparse sampling based multiple parameter estimation,''
\textit{Signal Processing.}, vol. 179, pp. 107825, 2021.

\bibitem{Xiao2023on}
H. S. Xiao, Y. W. Zhang, B. N. Zhou and G. Q. Xiao ``On the foundation of sparsity constrained sensing-Part I: Sampling theory and robust remainder problem,''
\textit{IEEE Transactions on Signal Processing.}, vol. 71, pp. 1263-1276, 2023.




\bibitem{Krim1996}
H. Krim and M. Viberg, ``Two decades of array signal processing research: The parametric approach,''
\textit{IEEE Signal Process. Mag.}, vol. 13, no. 4, pp. 67-94, Jul. 1996.

\bibitem{Godara1997}
L. C. Godara, ``Application of antenna arrays to mobile communications. II: Beam-forming and direction-of-arrival estimation,''
\textit{in Proc. IEEE}, vol. 85, no. 8, pp. 1195-1245, Aug. 1997.

\bibitem{Tuncer2009}
T. E. Tuncer and B. Friedlander, ``Classical and Modern Direction-of Arrival Estimation,''
New York, NY, USA: Academic, 2009.

\bibitem{Xiao22023}
H. S. Xiao, J. Wan and S. Devadas, ``Geometry of Sensitivity: Twice Sampling and Hybrid Clipping in Differential Privacy with Optimal Gaussian Noise and Application to Deep Learning,''
\textit{Proceedings of the 2023 ACM SIGSAC Conference on Computer and Communications Security.}, pp. 2636-2650, 2023.

\bibitem{Schmidt1986}
R. O. Schmidt, ``Multiple emitter location and signal parameter estimation,''
\textit{IEEE Trans. Antennas Propag.}, vol. AP-34, no. 3, pp. 276-280, Mar. 1986.


\bibitem{Roy1989}
R. Roy and T. Kailath, ``ESPRIT-estimation of signal parameters via rotational invariance techniques,''
\textit{IEEE Trans. Acoust., Speech, Signal Process.}, vol. 37, no. 7, pp. 984-995, Jul. 1989.

\bibitem{BD2017mutual}
E. BouDaher, F. Ahmad, M. G. Amin, and A. Hoorfar, ``Mutual coupling effect and compensation in non-uniform arrays for direction-of-arrival estimation,''
\textit{Digit. Signal Process.}, vol. 61, pp. 3-14, Feb. 2017.

\bibitem{Pal2010}
P. Pal and P. P. Vaidyanathan, ``Nested arrays: A novel approach to array processing with enhanced degrees of freedom,''
\textit{IEEE Trans. Signal Process.}, vol. 58, no. 8, pp. 4167-4181, Aug. 2010.

\bibitem{Pal2011}
P. P. Vaidyanathan and P. Pal, ``Sparse sensing with co-prime samplers and arrays,''
\textit{IEEE Trans. Signal Process.}, vol. 59, no. 2, pp. 573-586, Feb. 2011.


\bibitem{Sharma2023}
U. Sharma and M. Agrawal, ``Third-Order Nested Array: An Optimal Geometry for Third-Order Cumulants Based Array Processing,''
\textit{IEEE Transactions on Signal Processing.}, vol. 71, pp. 2849-2862, Aug. 2023.


\bibitem{Xiao2023}
H. S. Xiao, B. N. Zhou, Y. W. Zhang and G. Q Xiao, ``On the foundation of sparsity constrained sensing-Part II: Diophantine sampling with arbitrary temporal and spatial sparsity,''
\textit{IEEE Transactions on Signal Processing.}, vol. 71, pp. 1277-1292, Feb. 2023.

\bibitem{Guo2024}
H. D. Guo, H. Chen, H. G. Lin, W. Liu, Q. Shen and G. Wang, ``A New Fourth-Order Sparse Array Generator Based on Sum-Difference Co-Array Analysis,''
\textit{in ICASSP 2024-2024 IEEE International Conference on Acoustics, Speech and Signal Processing (ICASSP).}, pp. 8486-8490, Mar. 2024.

\bibitem{Moffet1968}
A. Moffet, ``Minimum-redundancy linear arrays,''
\textit{IEEE Trans. Antennas Propag.}, vol. 16, no. 2, pp. 172-175, Mar. 1968.

\bibitem{Pal22011}
P. Pal and P. P. Vaidyanathan, ``Coprime sampling and the MUSIC algorithm,''
\textit{in Proc. IEEE Digit. Signal Process. Signal Process. Educ. Meeting.}, pp. 289-294, Mar. 2011.


\bibitem{Zhao2019}
P. Zhao, G. Hu, Z. Qu, and L. Wang, ``Enhanced nested array configuration with hole-free co-array and increasing degrees of freedom for DOA estimation,''
\textit{IEEE Commun. Lett.}, vol. 23, no. 12, pp. 2224-2228, Dec. 2019.

\bibitem{Zheng2019}
Z. Zheng, W. Q. Wang, Y. Kong, and Y. D. Zhang, ``MISC array: A new sparse array design achieving increased degrees of freedom and reduced mutual coupling effect,''
\textit{IEEE Trans. Signal Process.}, vol. 67, no. 7, pp. 1728-1741, Apr. 2019.

\bibitem{Shi2022}
W. Shi, Y. Li, and R. C. de Lamare, ``Novel sparse array design based on the maximum inter-element spacing criterion,''
\textit{IEEE Signal Process. Lett.}, vol. 29, pp. 1754-1758, 2022.

\bibitem{Liu2015}
C. L. Liu and P. P. Vaidyanathan, ``Remarks on the spatial smoothing step in coarray MUSIC,''
\textit{IEEE Signal Process. Lett.}, vol. 22, no. 9, pp. 1438-1442, Sep. 2015.

\bibitem{Wang2019}
X. Wang and X. Wang, ``Hole identification and filling in k-times extended co-prime arrays for highly efficient DOA estimation,''
\textit{IEEE Trans. Signal Process.}, vol. 67, no. 10, pp. 2693-2706, May. 2019.


\bibitem{Shen2016}
Q. Shen, W. Liu, W. Cui, and S. Wu, ``Extension of nested arrays with the fourth-order difference co-array enhancement,''
\textit{in Proc. IEEE Int. Conf. Acoust., Speech, Signal Process.}, pp. 2991-2995, Mar. 2016.

\bibitem{Piya2012}
P. Pal and P. P. Vaidyanathan, ``Multiple level nested array: An efficient geometry for 2qth order cumulant based array processing,''
\textit{IEEE Transactions on Signal Processing.}, vol. 60, no.3, pp. 1253-1269, 2012.



\bibitem{LiuCL2016}
C. L. Liu and P. P. Vaidyanathan, ``Super nested arrays: Linear sparse arrays with reduced mutual coupling-Part I: Fundamentals,''
\textit{IEEE Trans. Signal Process.}, vol. 64, no. 15, pp. 3997-4012, Aug. 2016.


\bibitem{Shen2019}
Q. Shen, W. Liu, W. Cui, S. L. Wu, and P. Pal, ``Simplified and enhanced multiple level nested arrays exploiting high-order difference co-arrays,''
\textit{IEEE Transactions on Signal Processing.}, vol. 67, no. 13, pp.3502-3515, May. 2019.



\bibitem{Shen2015}
S. Qin, Y. D. Zhang and M. G. Amin, ``Generalized coprime array configurations for direction-of-arrival estimation,''
\textit{IEEE Trans. Signal Process.}, vol. 63, no. 6, pp. 1377-1390, Mar. 2015.

\bibitem{Cohen2019}
R. Cohen and Y. C. Eldar, ``Sparse fractal array design with increased degrees of freedom,''
\textit{in Proc. IEEE Int. Conf. Acoust., Speech, Signal Process.}, pp. 4195-4199, 2019.

\bibitem{Cohen2020}
R. Cohen and Y. C. Eldar, ``Sparse array design via fractal geometries,''
\textit{IEEE Trans. Signal Process.}, vol. 68, pp. 4797-4812, 2020.

\bibitem{Shen20152}
Q. Shen, W. Liu, W. Cui, S. Wu, Y. D. Zhang, and M. G. Amin, ``Low complexity direction-of-arrival estimation based on wideband co-prime arrays,''
\textit{IEEE/ACM Trans. Audio, Speech, Lang. Process.}, vol. 23, no. 9, pp. 1445-1456, Sep. 2015.

\bibitem{Shen2017}
Q. Shen, W. Liu, W. Cui, S. Wu, Y. D. Zhang, and M. G. Amin, ``Focused compressive sensing for underdetermined wideband DOA estimation exploiting high-order difference coarrays,''
\textit{IEEE Signal Process. Lett.}, vol. 24, no. 1, pp. 86-90, Jan. 2017.

\bibitem{Cui2019}
W. Cui, Q. Shen, W. Liu, and S. Wu, ``Low complexity DOA estimation for wideband off-grid sources based on re-focused compressive sensing with dynamic dictionary,''
\textit{IEEE J. Sel. Top. Signal Process.}, vol. 13, no. 5, pp. 918-930, Sep. 2019.

\bibitem{Zhou2018}
C. Zhou, Y. Gu, X. Fan, Z. Shi, G. Mao, and Y. D. Zhang, ``Direction-of arrival estimation for coprime array via virtual array interpolation,''
\textit{IEEE Trans. Signal Process.}, vol. 66, no. 22, pp. 5956-5971, Nov. 2018.

\bibitem{Liu20172}
C. L. Liu and P. P. Vaidyanathan, ``Maximally economic sparse arrays and cantor arrays,''
\textit{in Proc. IEEE 7th Int. Workshop Comput. Adv. MultiSensor Adaptive Process.}, pp. 1-5, 2017.

\bibitem{Robin2017}
R. Rajamki and V. Koivunen, ``Sparse linear nested array for active sensing.''
\textit{2017 25th European Signal Processing Conference (EUSIPCO).}, IEEE, 2017.


\bibitem{Robin2021}
R. Rajamki, and V. Koivunen, ``Sparse symmetric linear arrays with low redundancy and a contiguous sum co-array,''
\textit{IEEE Transactions on Signal Processing.}, vol. 69, no. , pp. 1697-1712, Feb. 2021.


\bibitem{Krentel1986}
J. B. Hiriart-Urruty, ``From convex optimization to nonconvex optimization. Necessary and sufficient conditions for global optimality,''
\textit{Nonsmooth optimization and related topics.}, pp. 219-239, 1989.

\bibitem{Yang2023}
Z. X. Yang, Q. Shen, W. Liu, Y. C. Eldar and W. Cui, ``High-order cumulants based sparse array design via fractal geometries Part I: Structures and DOFs,''
\textit{IEEE Transactions on Signal Processing.}, vol. 71, pp. 327-342, Feb. 2023.

%
%

\bibitem{Ahmed2017}
A. Ahmed, Y. D. Zhang and B. Himed, ``Effective nested array design for fourth-order cumulant-based DOA estimation,''
\textit{in Proc. IEEE Radar Conf.}, pp. 0998-1002, May. 2017.

\bibitem{Zhou2020}
Y. Zhou, Y. Li, L. Wang, C. Wen and W. Nie, ``The compressed nested array for underdetermined DOA estimation by fourth-order difference coarrays,''
\textit{in Proc. IEEE Int. Conf. Acoust., Speech, Signal Process., (ICASSP).}, pp. 4617-4621, May. 2020.


\bibitem{Friedlander1991}
B. Friedlander and A. J.Weiss, ``Direction finding in the presence of mutual coupling,''
\textit{IEEE Trans. Antennas Propag.}, vol. 39, no. 3, pp. 273-284, Mar. 1991.

\bibitem{Hickman1980}
J. H. Hickman, ``A note on the concept of multiset,''
\textit{Bulletin of the Australian Mathematical society}, vol.22,no. 2, pp. 211-217, 1980.

\bibitem{Svantesson19991}
T. Svantesson, ``Modeling and estimation of mutual coupling in a uniform linear array of dipoles,''
\textit{in Proc. IEEE Int. Conf. Acoust., Speech, Signal Process.}, pp. 2961-2964, Mar. 1999.

\bibitem{Svantesson19992}
T. Svantesson, ``Direction finding in the presence of mutual coupling,''
\textit{Masters thesis, Dept. Signals Syst., Chalmers Univ. Technol.}, 1999.

\bibitem{YeZ2009}
Z. Ye, J. Dai, X. Xu and X. Wu, ``DOA estimation for uniform linear array with mutual coupling,''
\textit{IEEE Trans. Aerosp. Electron. Syst.}, vol. 45, no. 1, pp. 280-288, Jan. 2009.

\bibitem{Liao2012}
B. Liao, Z. G. Zhang, and S. C. Chan, ``DOA estimation and tracking of ULAs with mutual coupling,''
\textit{IEEE Trans. Aerosp. Electron. Syst.}, vol. 48, no. 1, pp. 891-905, Jan. 2012.

\bibitem{Dias2017}
U. V. Dias and S. Srirangarajan, ``Co-prime arrays and difference set analysis,''
\textit{in Proc. 25th Eur. Signal Process. Conf.}, pp. 931-935, Aug. 2017.

\bibitem{Haan2007}
L. d. Haan and T. Koppelaars, ``Applied Mathematics for Database Professionals,''
New York, NY, USA: Apress, 2007.


\bibitem{LiuJ2017}
J. Liu, Y . Zhang, Y . Lu, S. Ren, and S. Cao, ``Augmented nested arrays with enhanced DOF and reduced mutual coupling,''
\textit{IEEE Trans. Signal Process.}, vol. 65, no. 21, pp. 5549-5563, Nov. 2017.

\bibitem{You2021}
C. Y. You and D. H. Zhang, ``Research on DOA Estimation Based on Spatial Smoothing Algorithm and Optimal Smoothing Times,''
\textit{IEEE International Conference on Data Science and Computer Application (ICDSCA).}, 2021.

\bibitem{Kane2024}
D. M. Kane, I. Diakonikolas, H. S. Xiao and S. H. Liu, ``Online Robust Mean Estimation,''
\textit{Proceedings of the 2024 Annual ACM-SIAM Symposium on Discrete Algorithms (SODA). Society for Industrial and Applied Mathematics.}, pp. 3197-3235, 2024.

%

%



\end{thebibliography}
\end{document}